\RequirePackage{fix-cm}
\documentclass[smallextended]{svjour3}       
\smartqed  
\usepackage{graphicx}
\usepackage{amssymb}
\usepackage{amsfonts}
\usepackage{geometry}
\usepackage{epsfig}
\usepackage{CJK}
\usepackage{cite}
\usepackage{amsmath}
\usepackage{epstopdf}
\usepackage{mathrsfs}
\usepackage{bbding}
\usepackage{cases}
\usepackage{subfigure}
\usepackage{float}
\usepackage{lineno,hyperref}
\usepackage{caption}
\usepackage{booktabs}
\usepackage{tabularx}
\everymath{\displaystyle}

\begin{document}

\title{ Complex excitations for the derivative nonlinear Schr\"{o}dinger equation
 \thanks{The project is supported by National Natural Science Foundation of China (No.12175069), Global Change Research Program of China (No.2015CB953904), Science and Technology Commission of Shanghai Municipality (No.18dz2271000 and No.21JC1402500).}}


\author{Huijuan Zhou \and Yong Chen \and Xiaoyan Tang \and Yuqi Li}


\institute{Huijuan Zhou \at
              School of Mathematical Sciences, Shanghai Key Laboratory of PMMP, East China Normal University, Shanghai, 200241, People’s Republic of China \\
           Yong Chen (\Envelope corresponding author)\at
               School of Mathematical Sciences, Shanghai Key Laboratory of PMMP, East China Normal University, Shanghai, 200241, People’s Republic of China\\[10pt]
  College of Mathematics and Systems Science, Shandong University of Science and Technology,
Qingdao, 266590, People’s Republic of China\\[10pt]
               Tel.: +86-21-62224199\\
              Fax: +86-21-62235025\\
              \email{profchenyong@163.com}\\
 Xiaoyan Tang \at
               School of Mathematical Sciences, Shanghai Key Laboratory of PMMP, East China Normal University, Shanghai, 200241, People’s Republic of China\\[10pt]\\
 Yuqi Li \at
               School of Mathematical Sciences, Shanghai Key Laboratory of PMMP, East China Normal University, Shanghai, 200241, People’s Republic of China\\[10pt]}
\date{Received: date / Accepted: date}
\maketitle
\begin{abstract}
\ \ \
The Darboux transformation (DT) formulae for the derivative nonlinear Schr\"{o}dinger (DNLS) equation are expressed in concise forms, from which the multi-solitons, $n$-periodic solutions, higher-order hybrid-pattern solitons and some mixed solutions are obtained. These complex excitations can be constructed thanks to more general semi-degenerate DTs. Even the non-degenerate $N$-fold DT with a zero seed can generate complicated $n$-periodic solutions.  It is proved that the solution $q[N]$ at the origin depends only on the summation of the spectral parameters. We find the maximum amplitudes of several classes of the wave solutions are determined by the summation.  Many interesting phenomena are discovered from these new solutions. For instance, the interactions between $n$-periodic waves produce peaks with different amplitudes and sizes; A soliton on a single periodic wave background shares a similar feature as a breather due to the interference of the periodic background. In addition, the results are extended to the reverse-space-time DNLS equation.
\end{abstract}

\section{Introduction}
\ \ \
The derivative nonlinear Schr\"{o}dinger (DNLS) equation 
\begin{equation}\label{dnls}
 iq_{t}-q_{x x}+i(q^{2}q^{*})_{x}=0,
\end{equation}
where the superscript $*$ denotes complex conjugation, is one of the most important integrable systems in mathematics and physics. It was first derived by various authors as an evolution equation for the change of an Alfv\'{e}n wave propagating along the magnetic field, when the weak nonlinearity and dispersion were taken into account\cite{rogister1971,Mjlhus1974,Mjlhus1976,Mio1976,Ichikawa1977}. The DNLS equation also governs the asymptotic state of the filamentation of lower-hybrid waves, which has moving solitary envelope solutions for the electric field \cite{Spatschek1977}. Ichikawa et al. \cite{iykk-jpsj-1980} revealed the peculiar structure of spiky modulations of amplitude and phase arises from the derivative nonlinear coupling term. In optics, it plays a significant role in the theory of ultrashort femtosecond nonlinear pulse \cite{cxj-pre-2004}. The formation of solitons on the sharp front of optical pulse in an optical fiber according to the DNLS equation was studied \cite{kam-pla-1998}.  The DNLS equation can also describe large-amplitude magnetohydrodynamic waves of plasmas \cite{plasma2}, the sub-picosecond and femtosecond pulses in a single-mode optical fiber \cite{7,8}. 
 
In \cite{Mjlhus1976}, Mj$\phi$lhus gave the exact solitary wave solutions as well as the conditions for modulational stability. In 1978, Kaup and Newell developed the inverse scattering techniques for the DNLS equation and obtained its $N$-soliton solutions \cite{kn-jmp-1978}. In their work, they also proposed the so-called Kaup-Newell (KN) system 
\begin{equation}\label{kns}
\begin{split}
 iq_{t}-q_{x x}-i(q^{2}r)_{x}=0,\\ 
 ir_{t}+r_{x x}-i(qr^{2})_{x}=0, 
\end{split}
\end{equation}
generalizing the DNLS equation. When $r(x,t)=-q^{*}(x,t)$, the KN system is reduced to the DNLS equation.  In 1999, Imai obtained the Darboux transformation (DT) for the DNLS equation and constructed the multi-soliton solutions, which include some new solutions, such as quasi-periodic solutions and soliton solutions under the quasi-periodic backgrounds \cite{kj-jpsj-1999}. In 2011, the rogue wave and rational traveling solution were obtained in \cite{JPAXSW-2011}. In 2012, the generated DT was established and thus the higher-order soliton solutions and rogue waves were obtained \cite{gbl-sam-2012}. The interactions between solitary waves have been widely studied in nonlinear science. In fact, perfect soliton interactions only happen in the lab, because in the complicated natural environment, the interaction process of solitons is often affected by other waves like the periodic waves. 
The periodic solutions are also important in theoretical researches. 
In general, it is extremely nontrivial to construct periodic solutions. The periodic and double-periodic solution expressed by complicated Jacobi elliptic functions have been constructed by many methods \cite{akhmediev-1987-tmp,hxr-2012-pre,cjb2018non}. Recently, the periodic solutions and double periodic solutions expressed by some trig functions were obtained by using DT with a plane wave seed \cite{lwhjs2018,zhj-nd-2021}. There are also some scattered periodic solution in various literatures in different disciplines \cite{txy-pre-2002}.  However, the study of $n$-periodic solutions is rarely considered. It is believed that the study of solitons on an $n$-periodic background can better meet the real situations.

This paper focuses on the multi-solitons (including velocity resonance solitons), $n$-periodic solutions, higher-order solitons, higher-order hybrid-pattern solitons and their mixed solutions for the DNLS equation. Here, we use the DT formulae to obtain these solutions. Most of these mixed solutions  were able to obtain because we constructed a more general semi-degenerate DTs in this work.
The paper is organized as follows. In Section 2, the $N$-fold DT formula with a concise expression without distinguishing the odd or even of $N$ is proposed, from which the $n$-periodic solutions, $n$-solitons and the mixed solutions are constructed. In Section 3, using the degenerate and semi-degenerate DT, the higher-order solitons and the mixed solutions of higher-order solitons and $n$-periodic waves are obtained. In Section 4, via the generalized degenerate and semi-degenerate DT, the higher-order hybrid-pattern solitons and the mixed solutions of higher-order hybrid-pattern solutions and $n$-periodic waves are constructed. In Section 5, the results are extended to the reverse-space-time DNLS equation and the modulational instability (MI) analysis is presented. The final section is devoted to conclusion and discussion.

{\bf \section{$n$-soliton, $n$-periodic solution and multi-soliton on the $n$-periodic background} }

A unique advantage of DT in solving integrable equations is to construct their solutions just by a pure algebraic procedure. In this section, the $N$-fold DT formula for the KN system is given in a concise form, which can be easily reduced to the DT formula for the DNLS equation. Then the multi-soliton, $n$-periodic solution and multi-soliton on the $n$-periodic background for the DNLS equation are obtained.

{\bf \subsection{$N$-fold DT for the KN system}}
The Lax pair of the KN system is known as follows \cite{kn-jmp-1978}
\begin{equation}
\begin{array}{c}\label{xlax}
\Psi_{x}=\left(i\sigma \lambda^{2}+Q \lambda\right) \Psi=U \Psi,\\
\end{array}
\end{equation}
\begin{equation}
\begin{array}{c}\label{tlax}
\Psi_{t}=\left(2 i\sigma \lambda^{4}+V_{3} \lambda^{3}+V_{2} \lambda^{2}+V_{1} \lambda\right) \Psi=V \Psi,
\end{array}
\end{equation}
with
$$
\Psi=\left(\begin{array}{c}
\phi \\
\varphi
\end{array}\right), \quad \sigma=\left(\begin{array}{cc}
1 &  0 \\
0 & -1
\end{array}\right), \quad Q=\left(\begin{array}{cc}
0 & \ q(x,t) \\
r(x,t) & \ 0
\end{array}\right),
$$
$$
V_{3}=2Q, \quad V_{2}=iQ^{2}, \quad V_{1}=Q^{3}+iQ_{x}\sigma,
$$
where the spectral parameter $\lambda$ is an arbitrary complex number.

Consider the gauge transformation
\begin{equation}\label{gt}
\Psi[N]=T \Psi,
\end{equation}
\begin{equation}\notag
T=\sum_{\ell=0}^{[\frac{N}{2}]}
   \left(\begin{array}{ll}
a_{N-2\ell}\lambda^{N-2\ell}\quad &  b_{N-2\ell-1}\lambda^{N-2\ell-1} \\
 c_{N-2\ell-1}\lambda^{N-2\ell-1}\quad & d_{N-2\ell}\lambda^{N-2\ell}
\end{array}\right),
\end{equation}
where $[\frac{N}{2}]$ is a least integer function, $a_{0}$, $b_{0}$, $c_{0}$ and $d_{0}$ are constants. When the subscripts of $a_{N-2\ell}$, $b_{N-2\ell-1}$, $c_{N-2\ell-1}$ and $d_{N-2\ell}$ are less than zero, the corresponding elements are zero, and the other elements of the matrix $T$ are functions of $x$ and $t$. Accordingly, the spectral problem \eqref{xlax} and \eqref{tlax} are transformed to
\begin{equation}\label{nlax}
\begin{array}{l}
\Psi[N]_{x}=U[N] \Psi[N], \quad U[N]=\left.U\right|_{q(x, t) \rightarrow q[N](x, t), r(x,t) \rightarrow r[N](x,t)}, \\
\Psi[N]_{t}=V[N] \Psi[N],  \quad V[N]=\left.V\right|_{q(x, t) \rightarrow q[N](x, t), r(x,t) \rightarrow r[N](x,t)},
\end{array}
\end{equation}
where ($q[N]$, $r[N]$) is the new solution of Eq. \eqref{kns}.

Based on Eqs. \eqref{gt} and \eqref{nlax}, the following results are obtained
\begin{equation}\label{3.3}
T_{x}=U[N]T-TU,
\end{equation}
\begin{equation}\label{3.4}
T_{t}=V[N]T-TV.
\end{equation}
Further, considering the kernel problem of the DT matrix $T$, i.e.,
\begin{equation}\label{hwt}
\left.T\right|_{\lambda=\lambda_{j}} \Psi_{j}=\sum_{\ell=0}^{[\frac{N}{2}]}
   \left(\begin{array}{ll}
a_{N-2\ell}\lambda_{j}^{N-2\ell}\quad &  b_{N-2\ell-1}\lambda_{j}^{N-2\ell-1} \\
 c_{N-2\ell-1}\lambda_{j}^{N-2\ell-1}\quad & d_{N-2\ell}\lambda_{j}^{N-2\ell}
\end{array}\right) \Psi_{j}=0, j=1,2, \cdots, N,
\end{equation}
where $\Psi_{j}=\left(\begin{array}{c}\phi_{j} \\ \varphi_{j}\end{array}\right)=\left(\begin{array}{c}
\phi_{j}\left(x, t, \lambda_{j}\right) \\
 \varphi_{j}\left(x, t, \lambda_{j}\right)
\end{array}\right)$ is the eigenfunction corresponding to $\lambda_{j}$.
 Then the concrete expression of the new solution $(q[N], r[N])$  can be seen in the following theorem by using the Cramer's rule and Eqs. \eqref{3.3}-\eqref{3.4}.

\begin{theorem}
{\bf ($N$-fold  Darboux transformation formula for the KN system):}
The solution $(q[N], r[N])$ for Eq. \eqref{kns} is given as
\begin{equation}\label{dtf}
q[N]=\frac{|M^{2}|}{|P^{2}|}q+2i\frac{|MH|}{|P^{2}|},
r[N]=\frac{|P^{2}|}{|M^{2}|}q+2i\frac{|PW|}{|M^{2}|},
\end{equation}
where $M=(M_{jk})_{1\leq j,k\leq N}$, $H=(H_{jk})_{1\leq j,k\leq N}$,
 $P=(P_{jk})_{1\leq j,k\leq N}$ and  $W=(W_{jk})_{1\leq j,k\leq N}$, with

\begin{equation}
\begin{split}
&M_{jk}=\left\{
         \begin{array}{ll}
           \lambda_{j}^{N-k}\varphi_{j}, & k=odd \\
          \lambda_{j}^{N-k}\phi_{j}, & k=even
         \end{array}
       \right.,\\
&H_{jk}=\left\{
         \begin{array}{ll}
           \lambda_{j}^{N}\phi_{j}, & k=1 \\
           \lambda_{j}^{N-k}\varphi_{j}, & \hbox{$k=even \geq 2$} \\
           \lambda_{j}^{N-k}\phi_{j}, & k=odd
         \end{array}
       \right.,\\
&P_{jk}=\left\{
         \begin{array}{ll}
           \lambda_{j}^{N-k}\phi_{j}, & k=odd \\
          \lambda_{j}^{N-k}\varphi_{j}, & k=even
         \end{array}
       \right.,\\
       &W_{jk}=\left\{
         \begin{array}{ll}
           \lambda_{j}^{N}\varphi_{j}, & k=1\\
           \lambda_{j}^{N-k}\phi_{j}, & \hbox{$k=even \geq 2$}\\
           \lambda_{j}^{N-k}\varphi_{j}, & k=odd
         \end{array}
       \right..\\
\end{split}
\end{equation}
\end{theorem}

\begin{remark}
Here, we present the $N$-fold DT formula with a more concise expression without distinguishing the parity of $N$. Owing to the reduction condition of the DT for the KN system, i.e., $r(x,t)=-q^{*}(x,t)$, the DT of the DNLS equation is established and the eigenfunctions of the DNLS equation satisfy the symmetry relation \cite{kj-jpsj-1999} $\phi^{*}(x,t,-\lambda_k^*)=\varphi(x,t,\lambda_j)$ only for $\lambda_{k}=-\lambda_{j}^{*}$. 
\end{remark}

In this study, we only consider the zero seed solution, namely, $q(x,t)=r(x,t)=0$. In this case, the eigenfunction of the Lax pair \eqref{xlax} and \eqref{tlax} is solved as
   \begin{equation}\label{sef}
\Psi_{j}=\left(\begin{array}{c}\phi_{j} \\ \varphi_{j}\end{array}\right)=\left(\begin{array}{c}e^{i(\lambda_{j}^{2} x+2 \lambda_{j}^{4} t)} \\ e^{-i(\lambda_{j}^{2}x+2 \lambda_{j}^{4}t)}\end{array}\right).
\end{equation}

The $n$-soliton solution, $n$-periodic solution and mixed solution of the  DNLS equation can be obtained  by the $N$-fold DT with the symmetry relation $\phi^{*}(x,t,-\lambda_k^*)=\varphi(x,t,\lambda_j)$ for $\lambda_{k}=-\lambda_{j}^{*}$. And the above solution can only be found if $N$ is even.

{\bf \subsection{$n$-soliton solution}}

The $n$-soliton solution can be constructed by substituting $\lambda_{2k}=-\lambda_{2k-1}^{*}=-\alpha_{2k-1}+i \beta_{2k-1}$, $k=1,..n,$ $n=\frac{N}{2}$ and the eigenfunction \eqref{sef} into the $N$-fold DT.
   
When $N=2$, the exact expression of the one-soliton solution $q_{s}[2]$ is 
  \begin{equation}
 \label{gzj}
q_{s}[2]=-4\alpha_{1}e^{iH}\frac{2\alpha_{1}\cosh{(F-iH)}-2\beta_{1}i\sinh{(F-iH)}}{(-2\beta_{1}i\sinh{(F-iH)}-m)^{2}},
\end{equation}
where
\begin{equation}
\begin{split}
&H=2(\alpha_{1}^{2}-\beta_{1}^{2})x+4(\beta_{1}^{4}-6\alpha_{1}^{2}\beta_{1}^{2}+\alpha_{1}^{4})t),\\
&F=16\alpha_{1}^{3}\beta_{1}t-16\alpha_{1}\beta_{1}^{3}t+4\alpha_{1}\beta_{1}x.
 \end{split}
\end{equation}
The center trajectory of  $q_{s}[2]$ is determined as $x=4(\beta_{1}^{2}-\alpha_{1}^{2})t$ and the amplitude of $q[2]$ is $4|\beta_{1}|$.  When $\alpha_{1} \rightarrow 0$, $q_{s}[2]$ becomes a rational soliton 
\begin{equation}
q[2]_{rs}=\frac{e^{2i\beta_{1}^{2}(2\beta_{1}^{2}t-x)}(64it\beta_{1}^{5}-16ix\beta_{1}^{3}-4\beta_{1})}
{1-256\beta_{1}^{8}t^{2}+128\beta_{1}^{6}tx+(32it-16x^{2})\beta_{1}^{4}-8i\beta_{1}^{2}x},
\end{equation}
where $\beta_{1}$ is an arbitrary real constant.
It is easy to find $|q[2]_{rs}|^{2}=\frac{16\beta_{1}^{2}}{(16\beta_{1}^{4}t-4\beta_{1}^{2}x)^{2}+1}$, which hints $q[2]_{rs}$ is an analytical solution at the whole $(x, t)$ plane and its trajectory is defined explicitly by $x=4\beta_{1}^{2}t$.

Since the exact expression of the multi-soliton solution is too complicated, we only give the dynamic evolutions as an illustration. We show the two-soliton, three-soliton and four-soliton solutions in Fig. \ref{2gzzt}. The velocity resonance $n$-solitons can be derived by $\beta_{2k-1}^{2}-\alpha_{2k-1}^{2}=v_{0}$, ($v_{0}$ is a constant), which are plotted in Fig. \ref{q2gzsdgzt}. The mixed solutions of the elastic collision solitons and velocity resonance solitons are shown in Fig. \ref{q83zt}.  According to the center trajectory equation of the soliton solution, the propagation direction and velocity of the soliton depend on $\alpha_{2k-1}$ and $\beta_{2k-1}$. The amplitude of soliton is a very important physical quantity, and hence more attention is paid to the amplitude of the $n$-soliton solution.

Taking a seed solution $q=0$ in \eqref{dtf}, the solution $q[N]$ at the origin $(0,0)$ reads 
$$q[N](0,0)=2i\frac{|M(0,0)H(0,0)|}{|P(0,0)^{2}|}=2i\frac{|H(0,0)|}{|P(0,0)|}
=-2i\sum_{j=1}^{N}\lambda_{j}.$$
In order to derive the amplitude of the $n$-soliton solution, taking $\lambda_{2k}=-\lambda_{2k-1}^{*}=-\alpha_{2k-1}+i \beta_{2k-1}$, $k=1,..n,$ $n=\frac{N}{2}$, then the soliton solution at the origin is $q_{s}[N](0,0)=4\sum_{k=1}^{n}\beta_{2k-1}$. Since the center trajectory of the $n$-soliton solution always passes through the origin, the amplitude of the $n$-soliton solution $q_{s}[N]$ is rightly $4|\sum_{k=1}^{n}\beta_{2k-1}|$ for $n=\frac{N}{2}$. The cross section diagrams of the solutions mentioned above are exhibited in Fig. \ref{gzjzft1}, which clearly demonstrate that the amplitude of the solutions is given by $4|\sum_{k=1}^{n}\beta_{2k-1}|$.
\begin{figure}[ht!]
\centering
\subfigure[]{
\label{2gz}
\begin{minipage}[b]{0.3\textwidth}
\includegraphics[width=3.0cm]{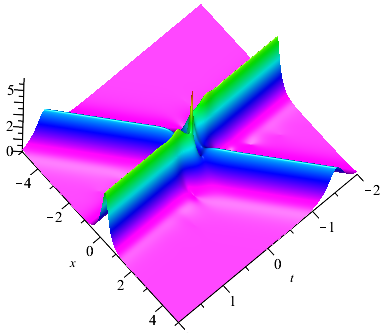}\\
\includegraphics[width=3.0cm]{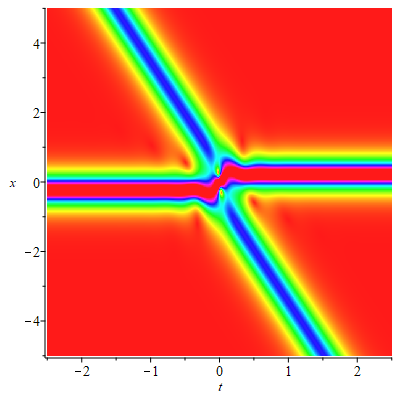}
\end{minipage}}
\subfigure[]{
\label{3gz}
\begin{minipage}[b]{0.3\textwidth}
\includegraphics[width=3.2cm]{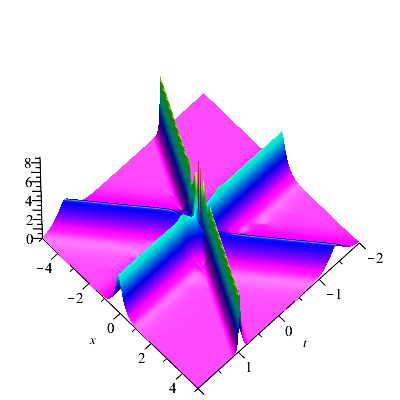}\\
\includegraphics[width=3.0cm]{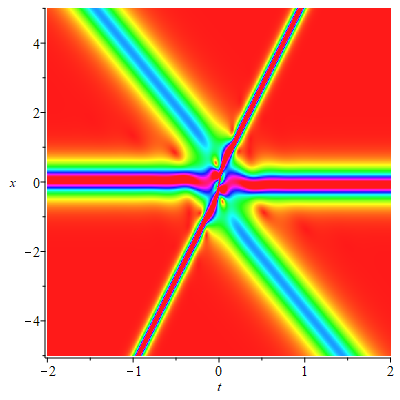}
\end{minipage}}
\subfigure[]{
\label{q81}
\begin{minipage}[b]{0.3\textwidth}
\includegraphics[width=3.2cm]{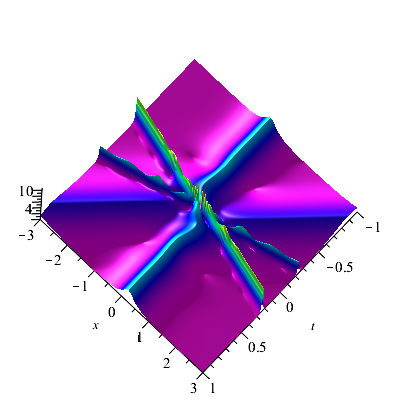}\\
\includegraphics[width=3cm]{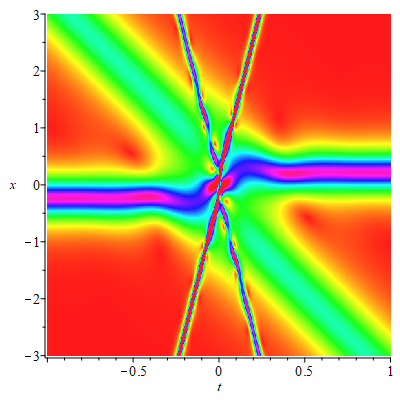}
\end{minipage}}
\caption{ 
(a) Two-soliton: $\alpha_{1}=\beta_{1}=\alpha_{3}=1$ and $\beta_{3}=0.5$; (b) Three-soliton: $\alpha_{1}=\beta_{1}=\alpha_{3}=\alpha_{5}=1$, $\beta_{3}=0.5$ and$\beta_{5}=2$;
  (c) Four-soliton: $\alpha_{1}=\beta_{1}=\alpha_{3}=\alpha_{5}=\beta_{7}=1$, $\beta_{3}=0.5$ and $\beta_{5}=\alpha_{7}=2$.}
\label{2gzzt}
\end{figure}

\begin{figure}[ht!]
\centering
\subfigure[]{
\label{q2gzsdg}
\begin{minipage}[b]{0.3\textwidth}
\includegraphics[width=3.4cm]{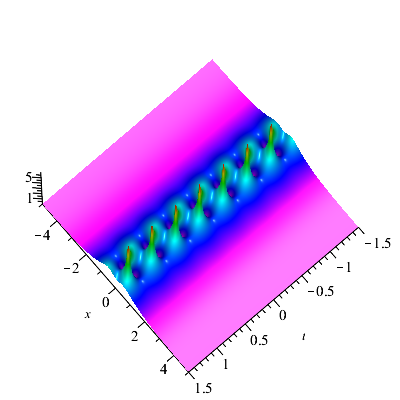}\\
\includegraphics[width=3.0cm]{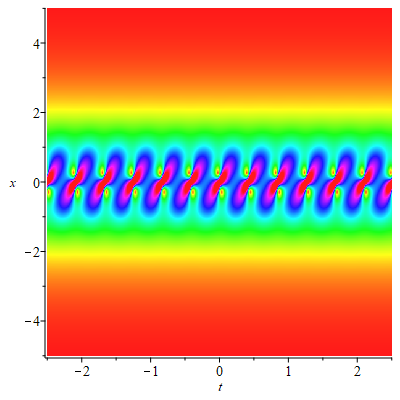}
\end{minipage}}
\subfigure[]{
\label{vc3gz}
\begin{minipage}[b]{0.3\textwidth}
\includegraphics[width=3.4cm]{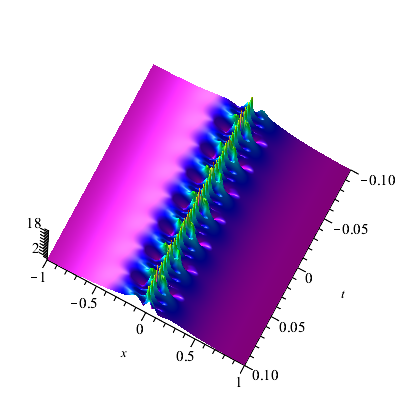}\\
\includegraphics[width=3.0cm]{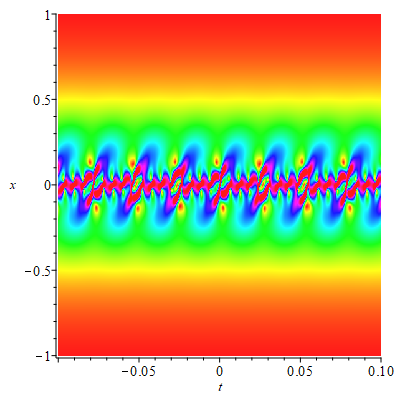}
\end{minipage}}
\subfigure[]{
\label{vc4gz}
\begin{minipage}[b]{0.3\textwidth}
\includegraphics[width=3.3cm]{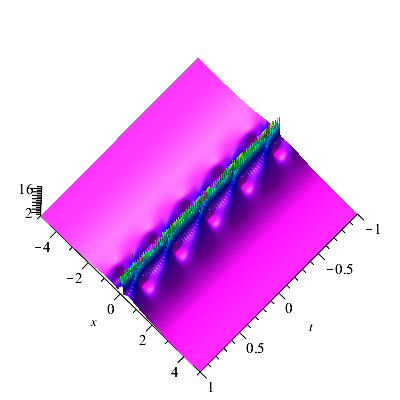}\\
\includegraphics[width=3cm]{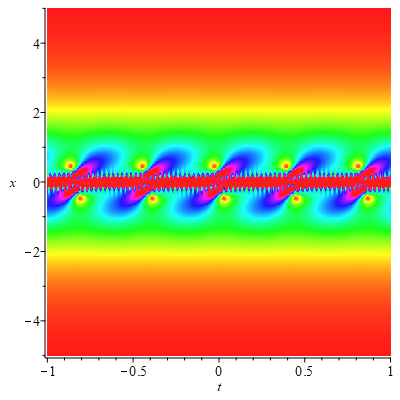}
\end{minipage}}
\caption{ (a) Velocity resonance two-soliton: $\alpha_{1}=\beta_{1}=1$ and $\alpha_{3}=\beta_{3}=0.5$; (b) Velocity resonance three-soliton: $\alpha_{1}=\beta_{1}=1$, $\alpha_{3}=\beta_{3}=2$ and $\alpha_{5}=\beta_{5}=3$; (c) Velocity resonance four-soliton:  $\alpha_{1}=\beta_{1}=1$, $\alpha_{3}=\beta_{3}=0.5$, $\alpha_{5}=\beta_{5}=2$ and $\alpha_{7}=\beta_{7}=3$.
}
\label{q2gzsdgzt}
\end{figure}

\begin{figure}[ht!]
\centering
\subfigure[]{
\label{n32gzgz1gz}
\begin{minipage}[b]{0.3\textwidth}
\includegraphics[width=3.3cm]{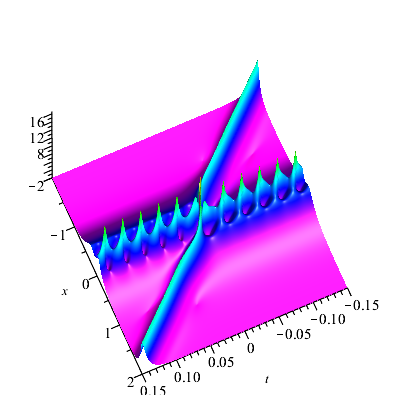}\\
\includegraphics[width=3.0cm]{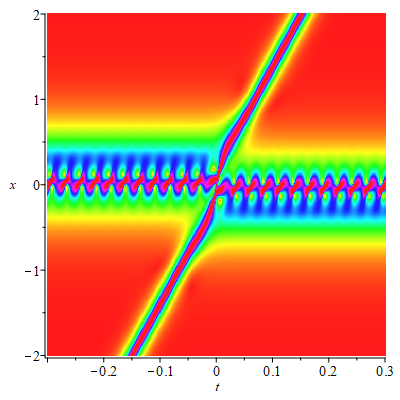}\\
\end{minipage}}
\subfigure[]{
\label{q82}
\begin{minipage}[b]{0.3\textwidth}
\includegraphics[width=3.3cm]{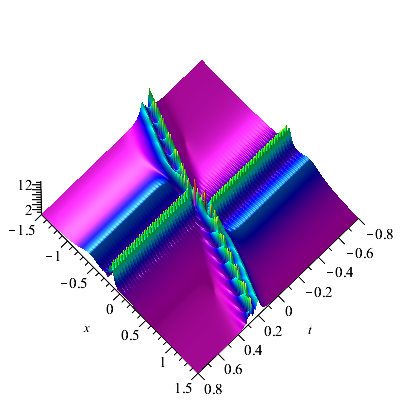}\\
\includegraphics[width=3cm]{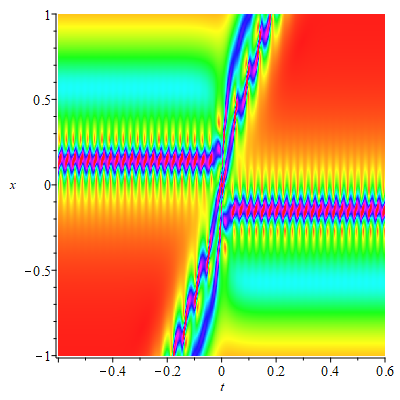}\\
\end{minipage}}
\subfigure[]{
\label{q83}
\begin{minipage}[b]{0.3\textwidth}
\includegraphics[width=3.3cm]{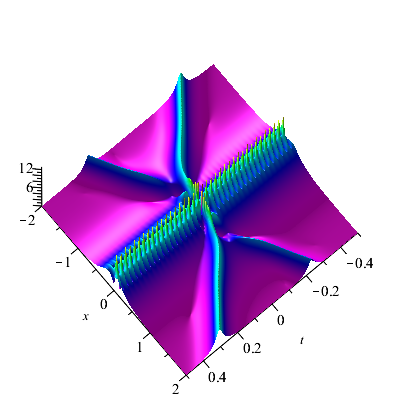}\\
\includegraphics[width=3cm]{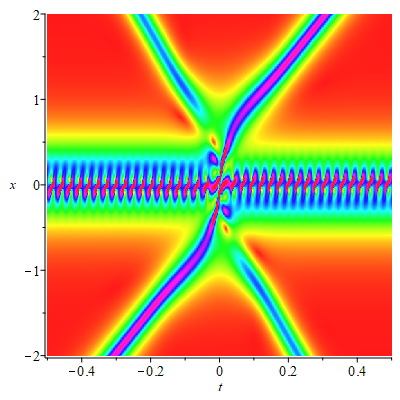}
\end{minipage}}
\caption{ (a) Elastic collision of one-soliton and velocity resonance two-soliton: $\alpha_{1}=\beta_{1}=\alpha_{5}=1$ and
$\alpha_{3}=\beta_{3}=\beta_{5}=2$; (b) Two velocity resonance two-soliton: $\alpha_{1}=\beta_{1}=\alpha_{5}=1$, $\alpha_{3}=\beta_{3}=\beta_{7}=2$, $\beta_{5}=1.5$ and $\alpha_{7}=\frac{\sqrt{11}}{2}$; (c) Elastic collision of two-soliton and velocity resonance two-soliton: $\alpha_{1}=\beta_{1}=\alpha_{5}=\beta_{7}=1$, $\alpha_{3}=\beta_{3}=2$,  $\beta_{5}=1.5$ and $\alpha_{7}=\frac{\sqrt{11}}{2}$.}
\label{q83zt}
\end{figure}

\begin{figure}[ht!]
\centering
\includegraphics[width=13cm]{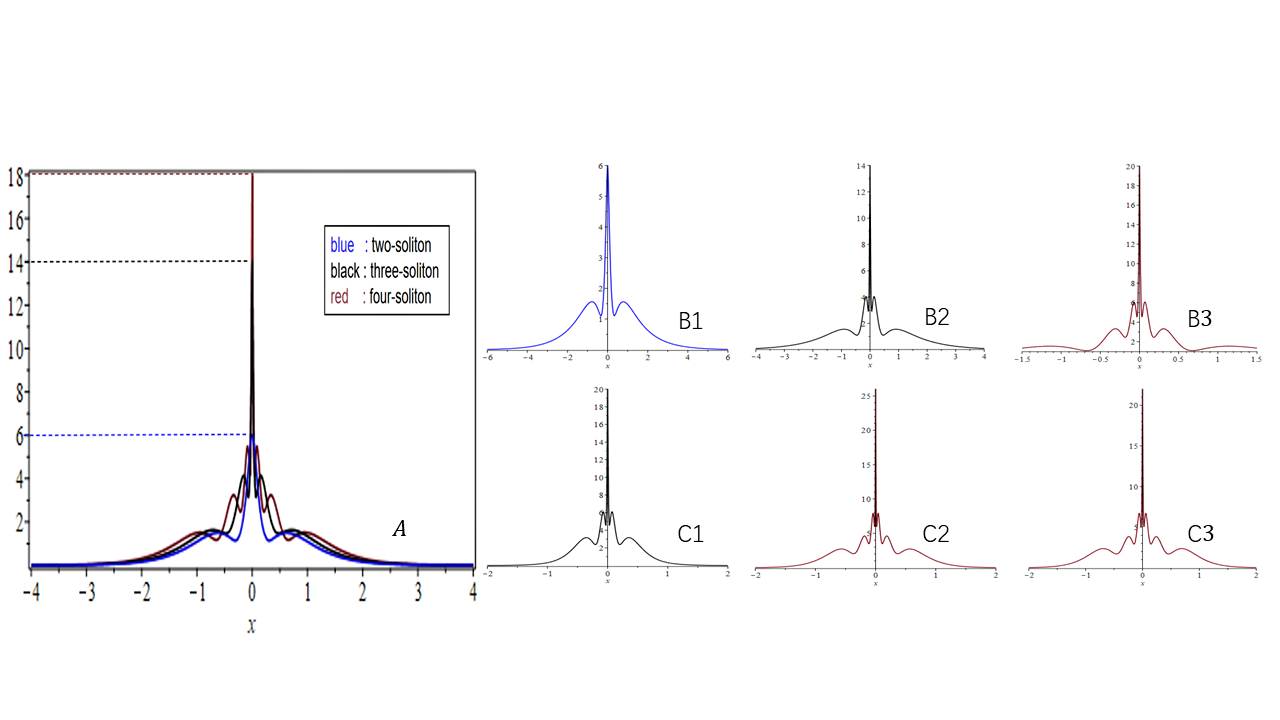}
 \caption{A: Cross section diagrams of two-soliton, three-soliton and four-soliton which plotted in Fig. \ref{2gzzt}; B1,B2 and B3: Cross section diagrams of  Velocity resonance two-soliton, three-soliton and four-soliton which plotted in Fig. \ref{q2gzsdgzt}; C1,C2 and C3: Cross section diagrams of elastic collision of one-soliton and velocity resonance two-soliton, two velocity resonance two-soliton and two-soliton and velocity resonance two-soliton which plotted in Fig. \ref{q83zt}. }
 \label{gzjzft1}
\end{figure}

{\bf \subsection{$n$-periodic solution}}

Taking $\lambda_{j}=i\beta_{j}$, $j=1,.. ,2n$, $2n=N$, the $n$-periodic solution can be constructed by the $N$-fold DT formula \eqref{dtf}. For example, letting $N=2$, we have the periodic (i.e. one-periodic) solution
\begin{equation}
\label{zqj}
q_{p}[2]=-2(\beta_{1}-\beta_{2})e^{i(K_{1}+K_{2})}\frac{i(\beta_{1}-\beta_{2})\cos{(K_{1}-K_{2})}-(\beta_{1}+\beta_{2})\sin{(K_{1}-K_{2})}}
{[-i(\beta_{1}-\beta_{2})\cos{(K_{1}+K_{2})}-(\beta_{1}+\beta_{2})\sin{(K_{1}-K_{2})}]^{2}},
\end{equation}
where $K_{1}=2\beta_{2}^{2}(2\beta_{2}^{2}t-x)$, $K_{2}=2\beta_{1}^{2}(2\beta_{1}^{2}t-x)$, and the period of $q_{p}[2]$ is $ \frac{\pi}{\beta_{1}^{2}-\beta_{2}^{2}}$. When $\beta_{1}$ and $\beta_{2}$ are of the same sign, the maximum and minimum values of  $q_{p}[2]$ are 
$2|\beta_{1} + \beta_{2}|$ and $2|\beta_{1}-\beta_{2}|$, respectively. Otherwise, the maximum value of  $q_{p}[2]$ is $2|\beta_{1} -\beta_{2}|$, and the minimum value is $2|\beta_{1}+\beta_{2}|$. That is to say, the amplitude and period of $q_{p}[2]$ rely on parameters 
$\beta_{1}$ and $\beta_{2}$. The periodic solution in Fig. \ref{q23} is plotted by taking $\beta_{1}=1$ and $\beta_{2}=0.5$. Intuitively, a one-periodic solution looks like a set of parallel solitons.

The dynamics of double-periodic and three-periodic solutions are displayed in Fig. \ref{szq2} and  Fig. \ref{q613zq}, respectively.  Double-periodic solutions are two sets of parallel solitons with different directions. A peak is created at the location where the periodic waves collide, so the double-periodic dynamic evolution looks like a set of parallel peaks with an equal amplitude as shown. The density diagram shows the double-periodic waves suffer phase shifts, very similar to the elastic collision of solitons. The dynamic evolution diagram of $n$-periodic solution ($n>2$) shows more complex structures, because the elastic collision of three periodic solutions with different directions and velocities produce peaks with different amplitudes and sizes. Due to the frequent collisions of periodic waves resulting in the frequent phase shift, the density graph of $n$-periodic wave contains irregular curves, as exhibited in Fig. \ref{q613zq}.

\begin{figure}[ht!]
\centering
\subfigure[]{
\label{q23}
\begin{minipage}[b]{0.3\textwidth}
\includegraphics[width=2.8cm]{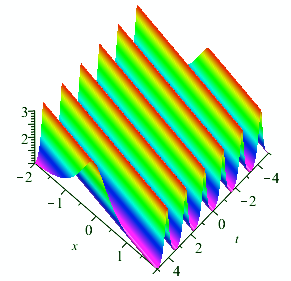}\\
\includegraphics[width=3.2cm]{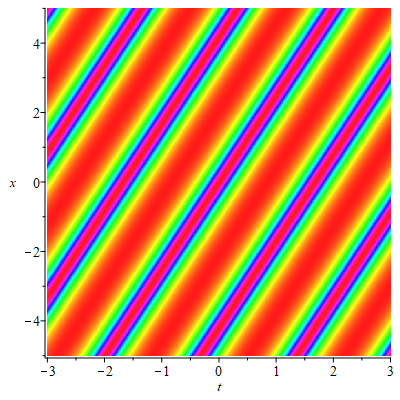}
\end{minipage}}
\subfigure[]{
\label{szq2}
\begin{minipage}[b]{0.3\textwidth}
\includegraphics[width=3.2cm]{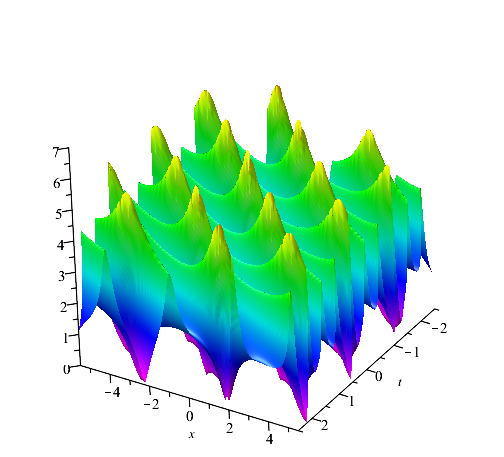}\\
\includegraphics[width=3.2cm]{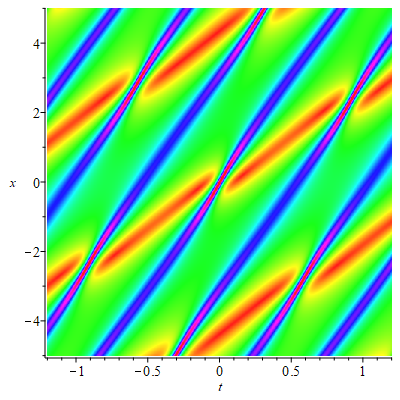}
\end{minipage}}
\subfigure[]{
\label{q613zq}
\begin{minipage}[b]{0.3\textwidth}
\includegraphics[width=3.2cm]{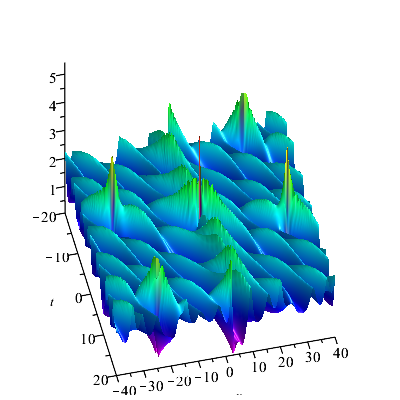}\\
\includegraphics[width=3.2cm]{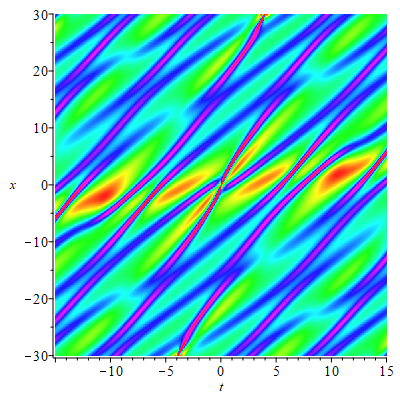}
\end{minipage}}
\caption{(a) Periodic solution: $\beta_{1}=1$ and $\beta_{2}=0.5$; (b) Double-periodic solution: $\beta_{1}=1$, $\beta_{2}=0.5$, $\beta_{3}=0.1$ and $\beta_{4}=\sqrt{2}$;  (c) Three-periodic solution: $\beta_{1}=0.1$, $\beta_{2}=0.9$, $\beta_{3}=0.2$, $\beta_{4}=0.8$, $\beta_{5}=0.3$ and  $\beta_{6}=0.7$.}
\label{nperiodic}
\end{figure}

In Fig. \ref{ztbjt}, we construct $n$-periodic solutions with different periods of the same amplitude and show the cross-section of $n$-periodic solutions for $n=1$, $2$, $3$, and $4$. When $\beta_{j}$ are of the same sign, the maximum amplitude of $n$-periodic wave solution $q_{p}[N]$ is $2|\sum_{j=1}^{2n}\beta_{j}|$ with $2n=N$. From the cross-sectional view, when $n$ is larger, the periodicity of the periodic wave is worse, and even the profile of a quasi rogue wave solution appears. The reason is that when $n$ increases, it is rare for $n$ single periodic waves to collide completely. In the figure below, we use different colors to represent $n$-periodic solutions with different parameters.

\begin{figure}[ht!]
\centering
\subfigure[]{
\label{x01zqbjt}
\begin{minipage}[b]{0.2\textwidth}
\includegraphics[width=3.1cm]{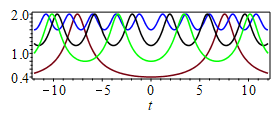}\\
\includegraphics[width=3.1cm]{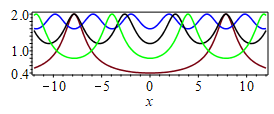}
\end{minipage}}
\subfigure[]{
\label{x02zqbjt}
\begin{minipage}[b]{0.2\textwidth}
\includegraphics[width=3.1cm]{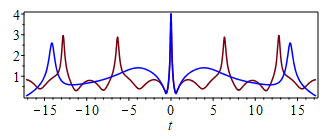}\\
\includegraphics[width=3.1cm]{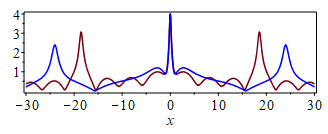}
\end{minipage}}
\subfigure[]{
\label{x03zqbjt}
\begin{minipage}[b]{0.2\textwidth}
\includegraphics[width=3.1cm]{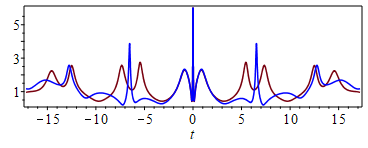}\\
\includegraphics[width=3.1cm]{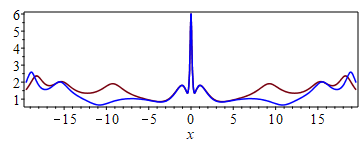}
\end{minipage}}
\subfigure[]{
\label{x04zqbjt}
\begin{minipage}[b]{0.2\textwidth}
\includegraphics[width=3.7cm]{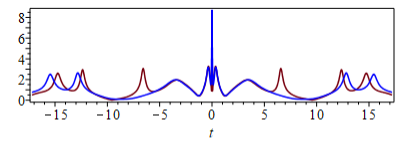}\\
\includegraphics[width=3.7cm]{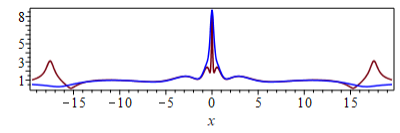}
\end{minipage}}
\caption{Periodic solution (a) Blue:$\beta_{1}=0.1$ and $\beta_{2}=-0.9$;
            Black:$\beta_{1}=0.2$ and $\beta_{2}=-0.8$;
            Green:$\beta_{1}=0.3$ and $\beta_{2}=-0.7$;
            Red:$\beta_{1}=0.4$ and $\beta_{2}=-0.6$;
 Double-periodic solution (b) Blue:$\beta_{1}=0.3$, $\beta_{2}= 0.7$, $\beta_{3}= 0.4$ and $\beta_{4}=0.6$;
 Red:$\beta_{1}=0.1$, $\beta_{2}=0.9$, $\beta_{3}=0.2$ and $\beta_{4}=0.8$;
Three-periodic solution (c) Blue:$\beta_{1} =0.1$, $\beta_{2} =0.9$, $\beta_{3} = 0.2$, $\beta_{4} = 0.8$, $\beta_{5} = 0.4$ and $\beta_{6} = 0.6$;
Red: $\beta_{1} = 0.1$, $\beta_{2} =0.9$, $\lambda_{3} = 0.2$, $\lambda_{4} = 0.8$, $\lambda_{5} = 0.3$ and $\lambda_{6} =0.7$;
Four-periodic solution (c) Blue: $\beta_{1} =0.1$, $\beta_{2} =0.9$, $\beta_{3} = 0.2$, $\beta_{4} = 0.8$, $\beta_{5} = 0.3$, $\beta_{6} =0.7$, $\beta_{7} = 0.4$ and $\beta_{8} = 0.6$;
Red: $\beta_{1} = 0.15$, $\beta_{2} =0.85$, $\beta_{3} = 0.2$, $\beta_{4} = 0.8$, $\beta_{5} = 0.3$, $\beta_{6} =0.7$,$\beta_{7} = 0.4$ and $\beta_{8} = 0.6$.
}
\label{ztbjt}
\end{figure}

{\bf \subsection{Multi-soliton on the $n$-periodic background}}

We can get the $m$-soliton on the $n$-periodic background by letting   $\lambda_{2k}$=$-\lambda_{2k-1}^{*}$=$-\alpha_{2k-1}+i \beta_{2k-1}$, $k=1,..,m.$ $\lambda_{j}=i\beta_{j}$, $j=2m+1$, .. ,$2m+2n$, $2(m+n)=N$ in the $N$-fold DT. Similar to the previous subsection, it is important to give here the formula for the maximum amplitude of $m$-solitons on the $n$-periodic background. When $\beta_{j}$ share the same sign, the maximum amplitude of  $m$-soliton on the $n$-periodic background is $4|\sum_{k=1}^{m}\beta_{2k-1}|+2\sum_{j=1}^{2n}|\beta_{j}|$. As an illustration, we just give the cases of $n=1$ and $2$.

{\bf Case n=1.} When $m=1$, one-soliton on the periodic background is constructed. It is observed from Fig. \ref{2} that a soliton on the periodic wave background shares a similar feature as a breather due to the interception of a periodic background. The period of the periodic background is adjusted by the values of the spectral parameters. Fig. \ref{2} is the soliton solution under different periodic backgrounds. Furthermore, it is seen that the local structure of a soliton on the periodic background has a single peak with two caves which is similar to a rogue wave. This implies rogue wave solutions can also be generated from zero seed solutions, but the feasibility remains to be proved. Then elastic collision of two solitons on a periodic background is constructed by setting $m=2$  (see Fig. \ref{n32gzp}). In particular, if $\beta_{2j-1}^{2}-\alpha_{2j-1}^{2}$=$v_{0}$  ($v_{0}$ is constant), $j=1$, $2$, then the velocity resonance of two solitons on a periodic-background is derived (see Fig. \ref{n6gzgzzq}). When $m=3$, the elastic collision of three-solitons on a periodic background is constructed (see Fig. \ref{q84}). In particular, if $\beta_{2j-1}^{2}-\alpha_{2j-1}^{2}$=$v_{0}$, $j=1$, $2$ and $v_{0}$ is constant, the elastic collision of  velocity resonance two-soliton and one-soliton on the periodic background takes place (see Fig. \ref{q851}).
 \begin{figure}[ht!]
\centering
\subfigure[]{
\label{n21gzp}
\includegraphics[width=7.5cm]{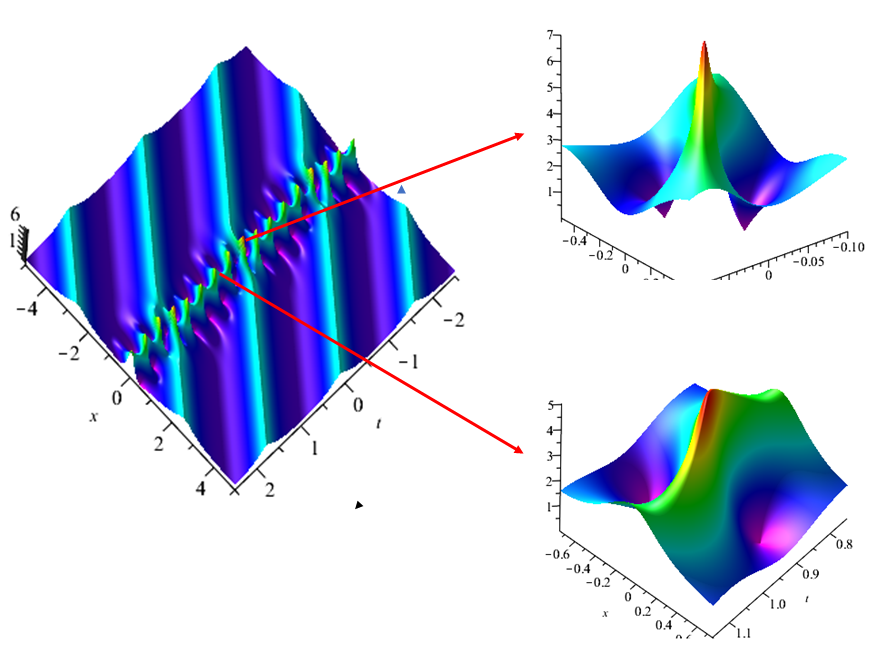}
\includegraphics[width=6.5cm]{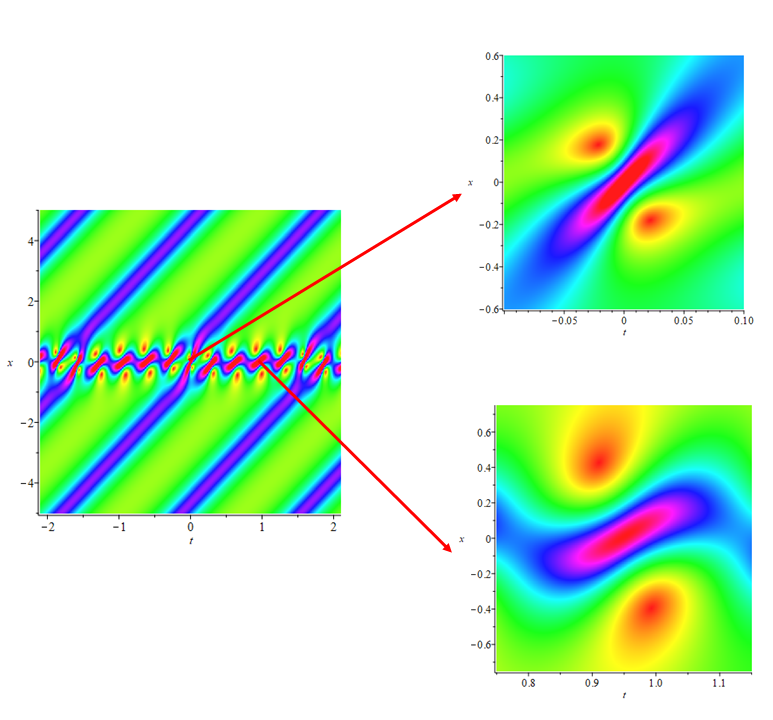}}
\subfigure[]{
\label{n21gzp2}
\includegraphics[width=4.0cm]{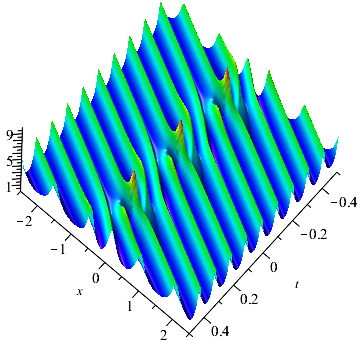}
\includegraphics[width=3.2cm]{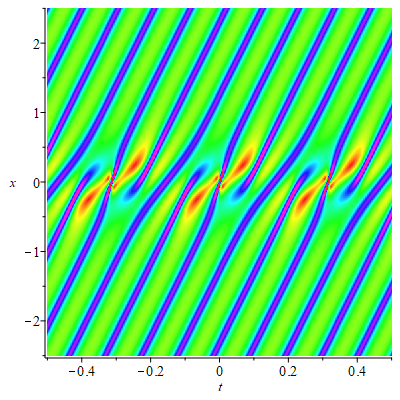}
\includegraphics[width=4.0cm]{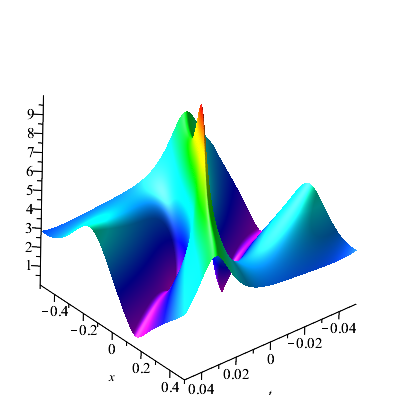}
\includegraphics[width=3.2cm]{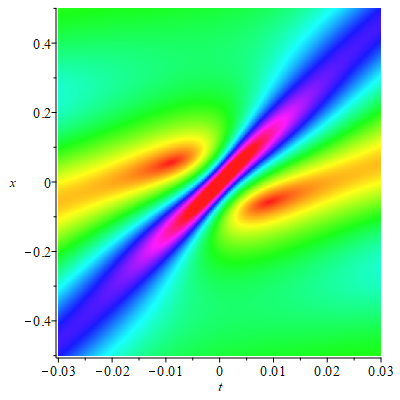}
}
\caption{ One-soliton on  different periodic backgrounds (a): $\alpha_{1}=\beta_{1}=\beta_{3}=1$ and $\beta_{4}=0.5$; (b): $\alpha_{1}=\beta_{1}=\beta_{3}=1$ and $\beta_{4}=2$.}
\label{2}
\end{figure}

\begin{figure}[ht!]
\centering
\subfigure[]{
\label{n32gzp}
\includegraphics[width=3.6cm]{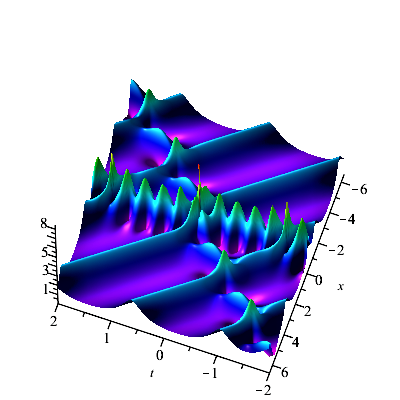}
\includegraphics[width=3.1cm]{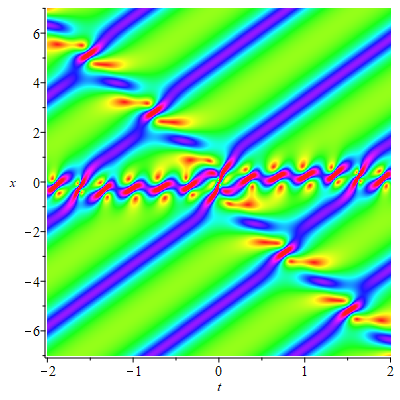}
\includegraphics[width=3.5cm]{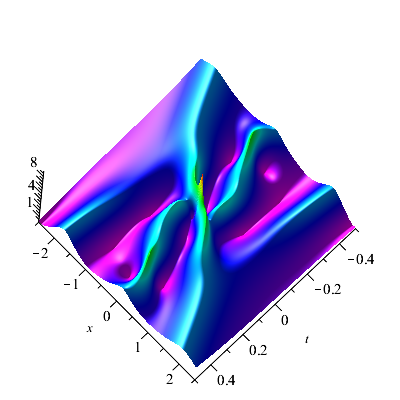}
\includegraphics[width=3.1cm]{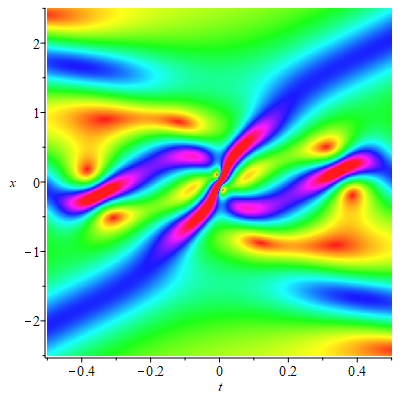}}
\subfigure[]{
\label{n6gzgzzq}
\includegraphics[width=3.5cm]{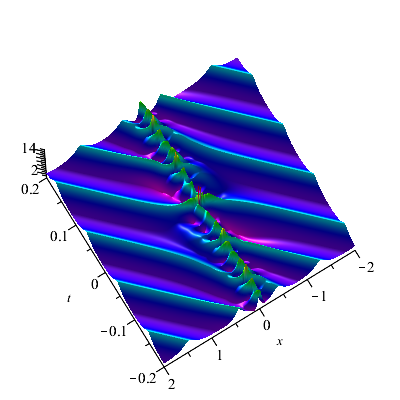}
\includegraphics[width=3.1cm]{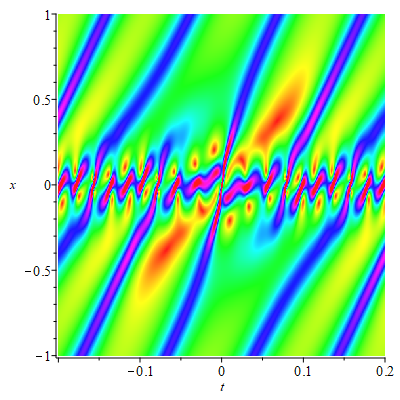}
\includegraphics[width=3.5cm]{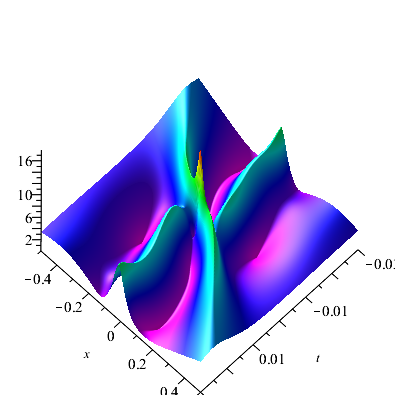}
\includegraphics[width=3.1cm]{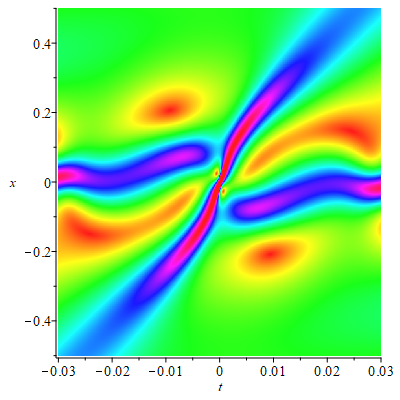}}
\caption{(a) Two-soliton on the periodic background: $\alpha_{1}=\beta_{1}=\alpha_{3}=\beta_{5}=1$ and $\beta_{3}=\beta_{6}=0.5$;
 (b) Velocity resonance two-soliton on the periodic background: $\alpha_{1}=\beta_{1}=\beta_{6}=1$ and $\alpha_{3}=\beta_{3}=\beta_{5}=2$.}
\end{figure}

\begin{figure}[ht!]
\centering
\subfigure[]{
\label{q84}
\begin{minipage}[b]{0.3\textwidth}
\includegraphics[width=3.3cm]{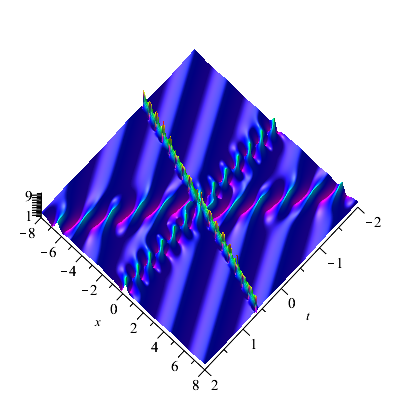}\\
\includegraphics[width=3cm]{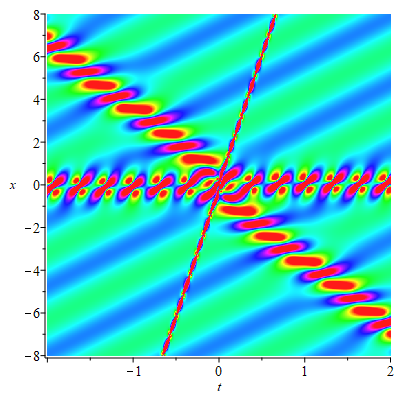}
\end{minipage}}
\subfigure[]{
\label{q851}
\begin{minipage}[b]{0.3\textwidth}
\includegraphics[width=3.3cm]{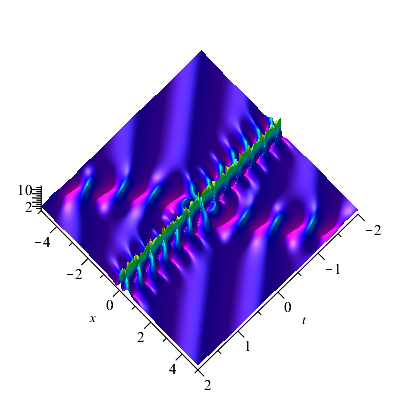}\\
\includegraphics[width=3cm]{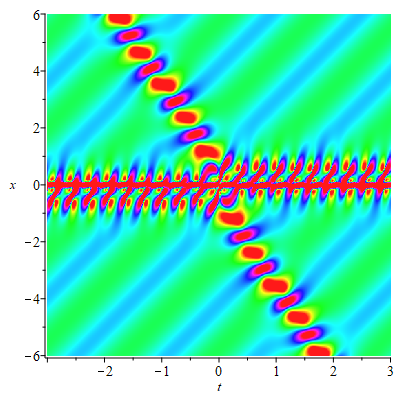}
\end{minipage}}
\subfigure[]{
\label{3vc1zq}
\begin{minipage}[b]{0.3\textwidth}
\includegraphics[width=3.3cm]{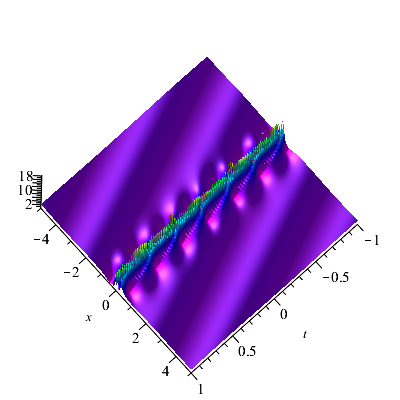}\\
\includegraphics[width=3cm]{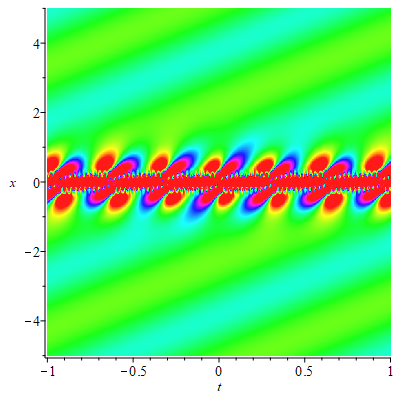}
\end{minipage}}
\caption{(a) Three soliton on the periodic background:
$\alpha_{1}=\beta_{1}=\alpha_{3}=\alpha_{5}=\beta_{8}=1$, $\beta_{3}=2$, $\beta_{5}=0.5$ and $\beta_{7}=0.1$; (b) Elastic collision of velocity resonance two-soliton and one-soliton on the periodic background: $\alpha_{1}=\beta_{1}=\alpha_{5}=\beta_{8}=1$, $\alpha_{3}=\beta_{3}=2$, $\beta_{5}=0.5$ and $\beta_{7}=0.1$.}
\end{figure}

 {\bf Case n=2.} One-soliton on the double-periodic background can be constructed by taking $n_{1}=1$. In order to see the structure of the soliton solution on the double-periodic background more clearly, we give a local magnification in the right of Fig. \ref{n6gzszq}. Moreover, if $n_{1}=2$, elastic collision of two-soliton on the double-periodic background can be constructed (see Fig. \ref{q86}). Particularly, if $\beta_{j}^{2}-\alpha_{j}^{2}$=$v_{0}$  ($j=1, 2$, $v_{0}$ is constant), then the velocity resonance of two-solitons on double-periodic background is derived (see Fig. \ref{q87}). The dynamic diagrams of these new solutions show very complex and interesting nonlinear structures, which are first presented in this investigation.

\begin{figure}[ht!]
\centering
\includegraphics[width=3.4cm]{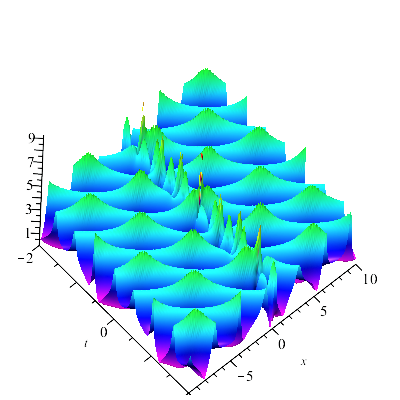}
\includegraphics[width=3cm]{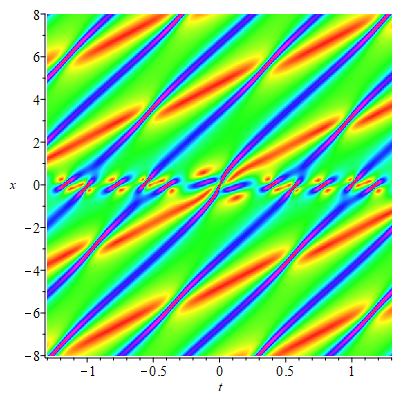}
\includegraphics[width=3.4cm]{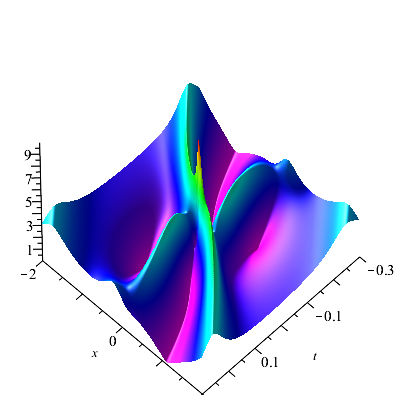}
\includegraphics[width=3cm]{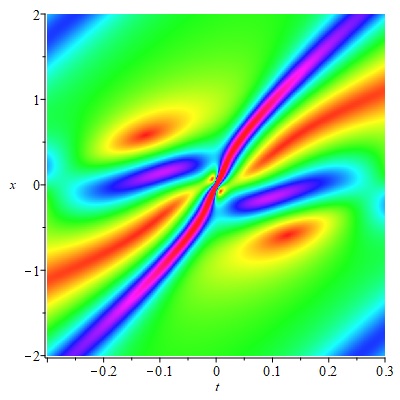}
\caption{ One-soliton on the double-periodic background: $\alpha_{1}=\beta_{1}=\beta_{3}=1$, $\beta_{4}=\beta_{5}=0.5$, $\beta_{5}=0.1$ and $\beta_{6}=\sqrt{2}$.}
\label{n6gzszq}
\end{figure}

\begin{figure}[ht!]
\centering
\subfigure[]{
\label{q86}
\includegraphics[width=3.5cm]{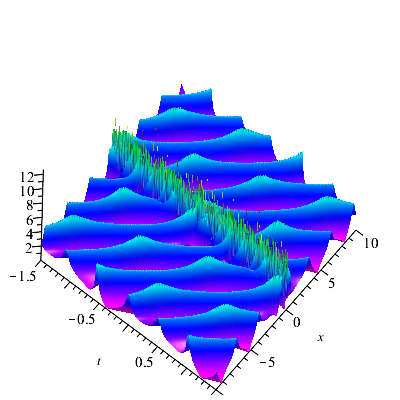}
\includegraphics[width=3cm]{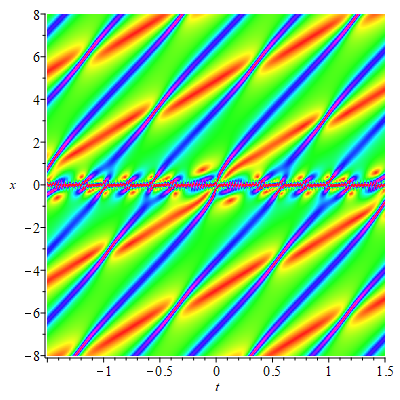}}
\subfigure[]{
\label{q87}
\includegraphics[width=3.5cm]{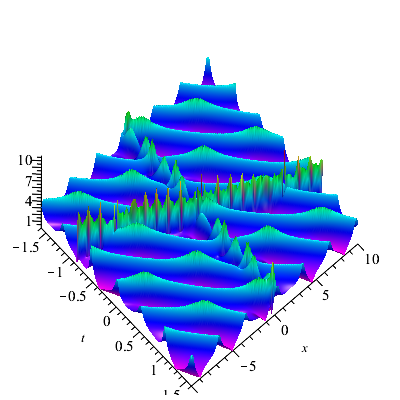}
\includegraphics[width=3cm]{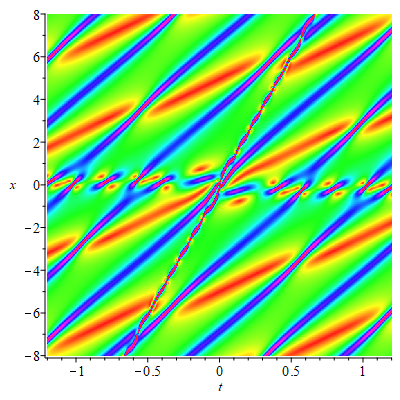}}
\caption{(a) Velocity resonance two-soliton on the double-periodic background: $\alpha_{1}=\beta_{1}=\beta_{5}=1$, $\alpha_{3}=\beta_{3}=2$,  $\beta_{6}=0.5$, $\beta_{7}=0.1$ and $\beta_{8}=\sqrt{2}$; (b) Elastic collision two-soliton on the double-periodic background: $\alpha_{1}=\beta_{1}=\alpha_{3}=\beta_{5}=1$, $\beta_{3}=2$,  $\beta_{6}=0.5$, $\beta_{7}=0.1$ and $\beta_{8}=\sqrt{2}$.}
\end{figure}


{\bf \section{Higher-order soliton without and with the $n$-periodic background}}

When considering iterations of DT, the same seed cannot be used twice. So, how to generate higher-order soliton is an interesting problem, which requires consider the nontrivial DT corresponding to $\lambda_{j}\rightarrow\lambda_{1}$ or $\lambda_{2}$, namely the degenerate DT. To construct higher-order soliton on the $n$-periodic background, we also need to construct semi-degenerate DT.

{\bf \subsection{Degenerate and semi-degenerate DT}}

The specific form of the degenerate and semi-degenerate DT is given below to construct the higher-order solutions and higher-order solutions with an $n$-periodic background.
\begin{theorem} {\bf (Degenerate and semi-degenerate DT):} Let $\lambda_{j}\rightarrow\left\{
                                             \begin{array}{ll}
                                               \lambda_{1}, & \hbox{$j=odd\leq N-2n$} \\
                                               \lambda_{2}, & \hbox{$j=even\leq N-2n$}
                                             \end{array}
                                           \right.$,
degenerate DT ($n=0$) and semi-degenerate DT ($n\neq0$) can be derived from the $N$-fold DT formula \eqref{dtf} by a Taylor expansion. It reads
\begin{equation}\label{ddtf}
q_{N}=\frac{\left|M'\right|^{2}}{|P'|^{2}}q+2i\frac{|M'||H^{'}|}{|P^{'}|^{2}},
\end{equation}
where
\begin{equation}
\begin{split}
&M'=\left\{
         \begin{array}{ll}
           M'_{jk}, & 1\leq j,k \leq N-2n \\
         M_{jk}, & N-2n<j,k \leq N
         \end{array}
       \right.,\\
&H'=\left\{
         \begin{array}{ll}
           H'_{jk}, & 1\leq j,k \leq N-2n \\
         H_{jk}, & N-2n <j,k \leq N
         \end{array}
       \right.,\\
&P'=\left\{
         \begin{array}{ll}
           P'_{jk}, & 1\leq j,k \leq n-2n \\
         P_{jk}, & n-2n< j,k \leq n
         \end{array}
       \right.,\\
\end{split}
\end{equation}

$$M_{jk}^{\prime}=\lim_{\epsilon \rightarrow 0}\frac{\partial^{n_{j}-1} M'_{jk}\left(\lambda_{j}+\epsilon\right)}{(n_{j}-1)! \partial \epsilon^{n_{j}-1}},$$

$$H_{jk}^{\prime}=\lim_{\epsilon \rightarrow 0}\frac{\partial^{n_{j}-1} H'_{jk}\left(\lambda_{j}+\epsilon\right)}{(n_{j}-1)! \partial \epsilon^{n_{j}-1}},$$

$$P_{jk}^{\prime}=\lim_{\epsilon \rightarrow 0}\frac{\partial^{n_{j}-1} P'_{jk}\left(\lambda_{j}+\epsilon\right)}{(n_{j}-1)! \partial \epsilon^{n_{j}-1}},$$

and $$n_{j}=\left[\frac{j+1}{2}\right].$$
\end{theorem}

\begin{proof}

As stated in the above, only the case where $N$ is an even number is considered, so we can take $N=2k$. Let $\lambda_{j}\rightarrow\left\{ 
                                             \begin{array}{ll}
                                               \lambda_{1}, & \hbox{$j$=odd} \\
                                               \lambda_{2}, & \hbox{$j$=even}
                                             \end{array}
                                           \right. $ in the formula \eqref{dtf}. 
                                           
First, perform the first order Taylor expansion in all the elements of the first and second rows with respect to $\varepsilon_{1}$. Extract $\varepsilon_{1}$ in the first and second rows, then take $\varepsilon_{1} \rightarrow 0$;

Second, take $\lambda_{2k-1-2n}=\lambda_{1}+\varepsilon_{2k-1-2n}$ and $\lambda_{2k-2n}=-\lambda_{1}+\varepsilon_{2k-1-2n}$, and do the 
$(k-n)$-th order Taylor expansion in all the elements of the $2k-1-2n$-th and $2k-2n$-th rows with respect to $\varepsilon_{2k-1-2n}$. Subtract the first, third, ... , $(2k-3-2n)$-th row from the $(2k-1-2n)$-th rows and subtract the second, fourth, ... , $(2k-2-2n)$-th row from the $2k-2n$-th rows. Extract $\varepsilon_{2k-1-2n}^{k-1-n}$ in the $(2k-1-2n)$-th rows and $(2k-2n)$-th rows, then take $\varepsilon_{2k-1-2n} \rightarrow 0$, where $\varepsilon_{2k-1-2n}$ is the real constant;

All the elements of the $(2k-2n+1)$-th to the  $2k$-th rows remain unchanged. Finally, the degenerate and semi-degenerate DT formula $q_{N}$ can be derived through determinant calculations.
\end{proof}

The higher-order soliton and higher-order soliton on the $n$-periodic background can be derived by the degenerate and semi-degenerate DT, respectively.

{\bf \subsection{Higher-order soliton derived by the degenerate DT}}

The higher-order soliton solution can be constructed by using the degenerate DT. Let us first consider the amplitude of the higher-order soliton solution. Taking a seed solution $q=0$ in Eq. \eqref{ddtf}, leads to
\begin{equation}\label{25}
q_{N}=2i\frac{|M'H'|}{|P^{'2}|}.
\end{equation}
The maximum amplitude of the higher-order soliton $q_{N}$ can be calculated at the origin $(0,0)$ without loss of generality. Obviously,
\begin{equation}
q_{N}(0,0)=2i\frac{|M'(0,0)H'(0,0)|}{|P'(0,0)^{2}|}=2i\frac{|H'(0,0)|}{|P'(0,0)|}
=-iN(\lambda_{1}+\lambda_{2}).
\end{equation}
Since the soliton is constructed for $\lambda_{2}=-\lambda_{1}^{*}=-\alpha_{1}+i\beta_{1}$, the  maximum amplitude of higher-order soliton $q_{N}$ in \eqref{25} is $2N|\beta_{1}|$.

For example, by setting $N = 4$, $6$ and $8$ in Theorem $2$,  the second, third and fourth-order solitons with $\alpha_{1}=1$ and $\beta_{1}=1$ are depicted in Fig. \ref{gjgz}. 
\begin{figure}[ht!]
\centering
\subfigure[]{
\label{2p}
\begin{minipage}[b]{0.22\textwidth}
\includegraphics[width=3.2cm]{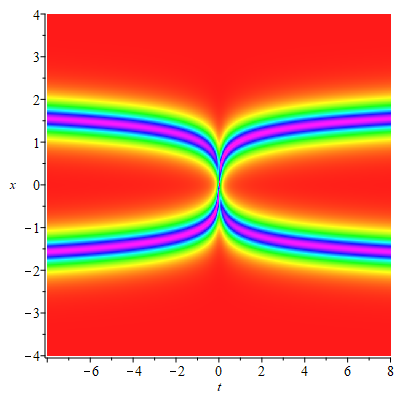}
\end{minipage}}
\subfigure[]{
\label{3p}
\begin{minipage}[b]{0.22\textwidth}
\includegraphics[width=3.2cm]{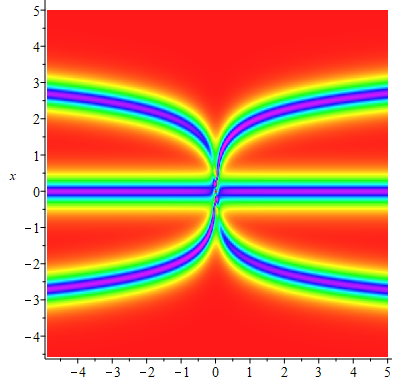}
\end{minipage}}
\subfigure[]{
\label{gq4}
\begin{minipage}[b]{0.22\textwidth}
\includegraphics[width=3.2cm]{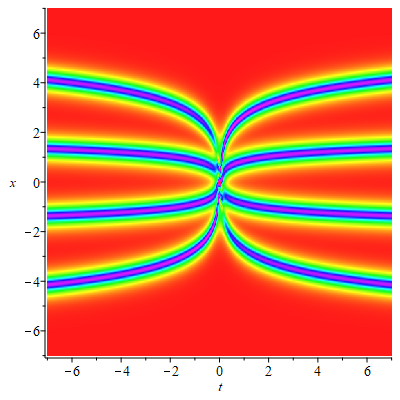}
\end{minipage}}
\subfigure[]{
\label{gjgz234}
\begin{minipage}[b]{0.22\textwidth}
\includegraphics[width=3.2cm]{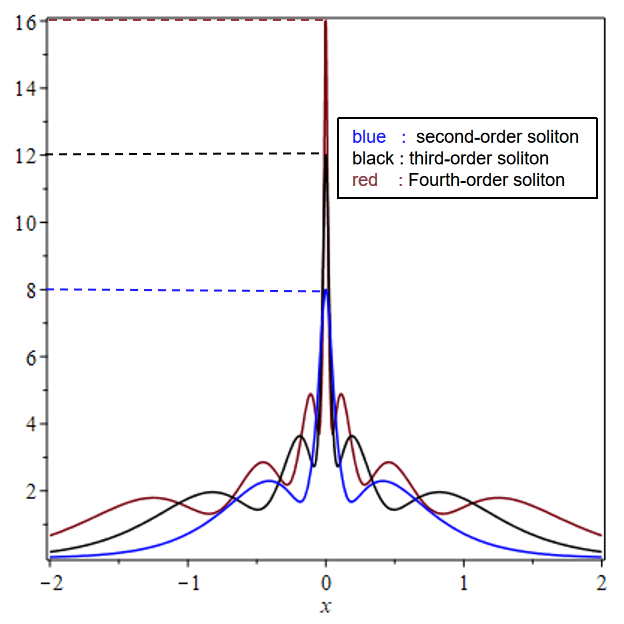}
\end{minipage}}
\caption{(a) Second-order soliton; (b) Third-order soliton; (c) Fourth-order soliton; (d) Cross section view of second-order, third-order and fourth-order solitons.}
\label{gjgz}
\end{figure}

{\bf \subsection{Higher-order soliton on the $n$-periodic background derived by the Semi-degenerate DT}}

 We can construct higher-order soliton on the $n$-periodic background by the semi-degenerate DT. When $\beta_{j}$ have the same sign, the maximum  height of the higher-order soliton on the multi-periodic background is
 $2(N-2n)|\beta_{1}|+2\sum_{j=N-2n+1}^{N}|\beta_{j}|$. As an application, we simply consider the cases of  $n=1$ and $2$.

In the case of $n=1$,  we show two types of solutions for $N=6$ and $8$, respectively. When  $N=6$, elastic collision between a second-order soliton and a one-soliton can be generated from Eq. \eqref{ddtf} by setting  $\lambda_{2}=-\lambda_{1}^{*}=-\alpha_{1}+i\beta_{1}$, $\lambda_{5}=\alpha_{5}+i\beta_{5}$ and $\lambda_{6}=-\lambda_{5}^{*}=-\alpha_{5}+i\beta_{5}$ (see Fig. \ref{g2+1q61}). In particular, when $\beta_{5}^{2}-\alpha_{5}^{2}$=$\beta_{1}^{2}-\alpha_{1}^{2}$, velocity resonance of second-order soliton and one soliton is derived (see Fig. \ref{g2+1q62}). If $\lambda_{5}=i\beta_{5}$ and $\lambda_{6}$=$i\beta_{6}$, the second-order soliton solution on the periodic background  is produced and the dynamic evolution is illustrated in Fig. \ref{2ppb}.
\begin{figure}[ht!]
\centering
\subfigure[]{
\label{g2+1q61}
\begin{minipage}[b]{0.3\textwidth}
\includegraphics[width=3.5cm]{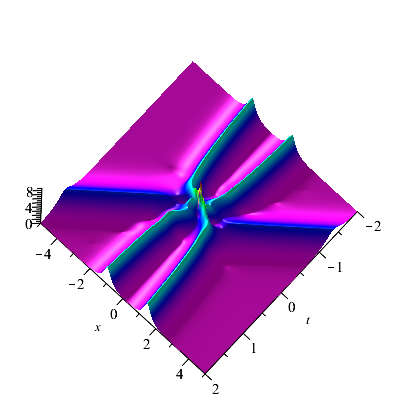}\\
\includegraphics[width=3.5cm]{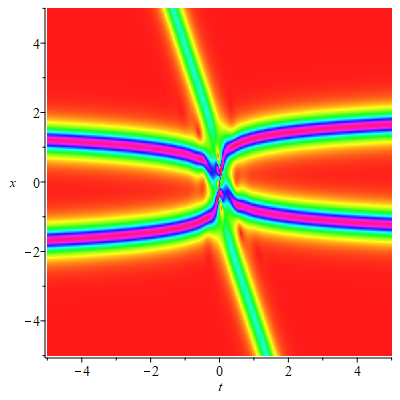}
\end{minipage}}
\subfigure[]{
\label{g2+1q62}
\begin{minipage}[b]{0.3\textwidth}
\includegraphics[width=3.5cm]{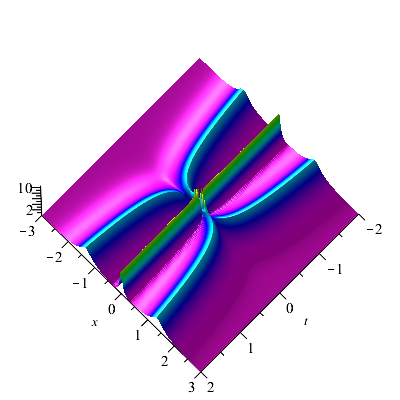}\\
\includegraphics[width=3.5cm]{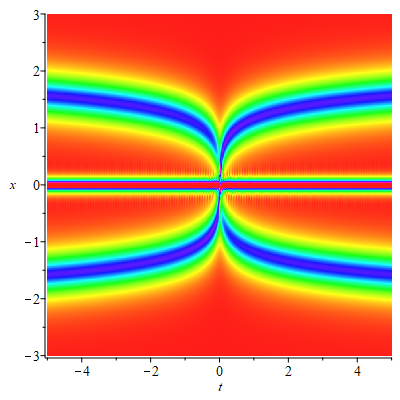}
\end{minipage}}
\subfigure[]{
\label{2ppb}
\begin{minipage}[b]{0.3\textwidth}
\includegraphics[width=3.5cm]{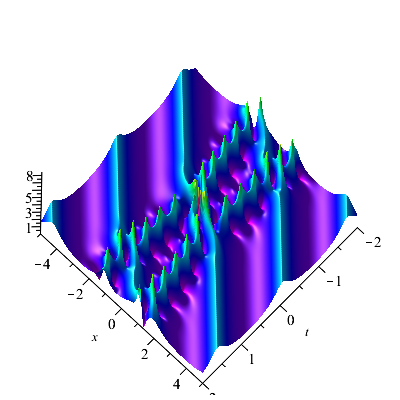}\\
\includegraphics[width=3.5cm]{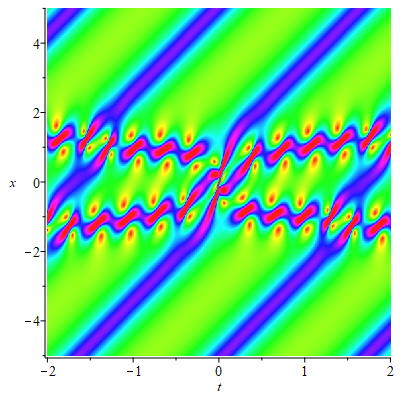}
\end{minipage}}
\caption{ (a) Elastic collision of second-order soliton and one soliton: $\alpha_{1}=\beta_{1}=\alpha_{5}=1$ and $\beta_{5}=0.5$; (b) Velocity resonance of second-order soliton and one soliton: $\alpha_{1}=\beta_{1}=1$ and $\alpha_{5}=\beta_{5}=2$; (c) Second-order soliton solution on the periodic background: $\alpha_{1}=\beta_{1}=\beta_{5}=1$ and $\beta_{6}=0.5$.}
\label{l1}
\end{figure}

When $N=8$, setting $\lambda_{2}=-\lambda_{1}^{*}=-\alpha_{1}+i\beta_{1}$, $\lambda_{8}=-\lambda_{7}^{*}=-\alpha_{7}+i\beta_{7}$, an elastic collision between a third-order soliton and a one-soliton is constructed. In particular, velocity resonance of third-order soliton and one-soliton is generated when $\beta_{7}^{2}-\alpha_{7}^{2}$=$\beta_{1}^{2}-\alpha_{1}^{2}$. A third-order soliton on a periodic background is derived by requiring $\lambda_{2}=-\lambda_{1}^{*}=-\alpha_{1}+i\beta_{1}$, $\lambda_{7}=i\beta_{7}, \lambda_{8}=i\beta_{8}$ and $\beta_{7}\neq\pm\beta_{8}$. The dynamical behaviours of these solutions can be seen in Fig. \ref{888}.
\begin{figure}[ht!]
\centering
\subfigure[]{
\label{g31ec}
\begin{minipage}[b]{0.3\textwidth}
\includegraphics[width=3.5cm]{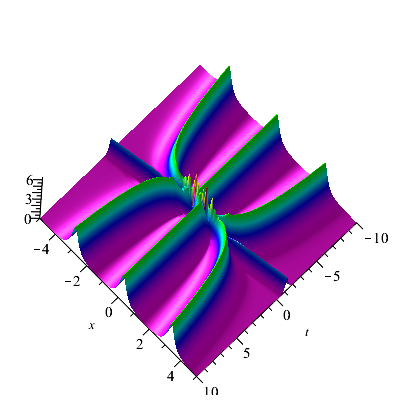}\\
\includegraphics[width=3.5cm]{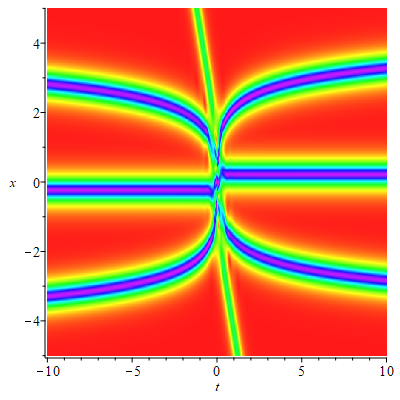}
\end{minipage}}
\subfigure[]{
\label{vc3j+1gz}
\begin{minipage}[b]{0.3\textwidth}
\includegraphics[width=3.5cm]{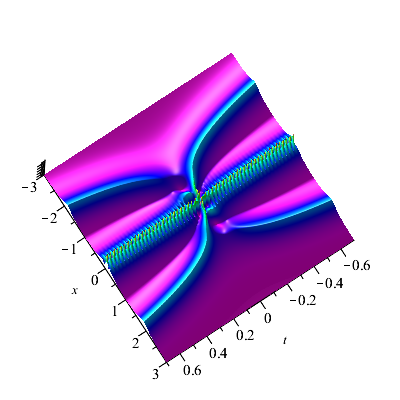}\\
\includegraphics[width=3.5cm]{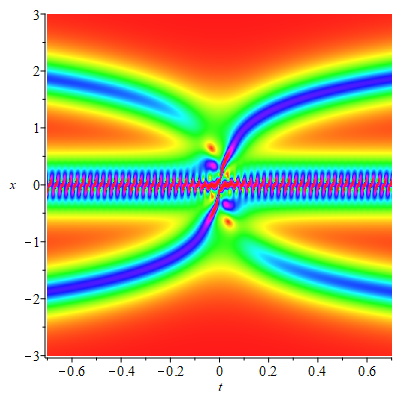}
\end{minipage}}
\subfigure[]{
\label{g3+1p}
\begin{minipage}[b]{0.3\textwidth}
\includegraphics[width=3.5cm]{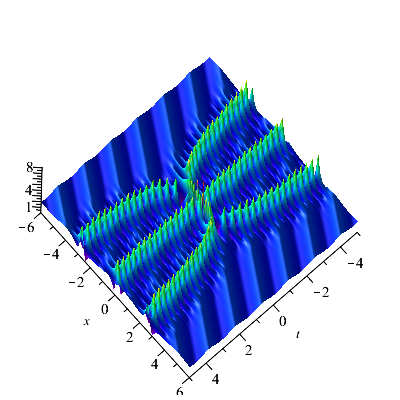}\\
\includegraphics[width=3.5cm]{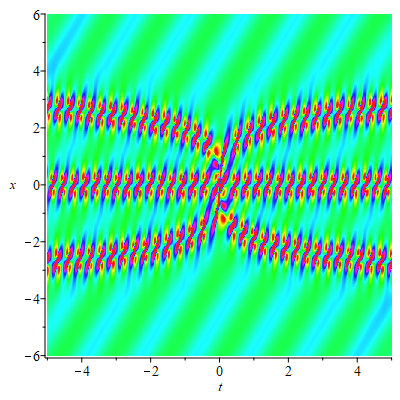}
\end{minipage}}
\caption{(a) Elastic collision of third-order soliton and one-soliton: $\alpha_{1}=\beta_{1}=\alpha_{7}=1$ and $\beta_{7}=0.5$;  (b) Velocity resonance of third-order soliton and one-soliton: $\alpha_{1}=\beta_{1}=1$ and $\alpha_{7}=\beta_{7}=2$; (c) Third-order soliton on the periodic background: $\alpha_{1}=\beta_{1}=\beta_{8}=1$ and $\beta_{7}=0.1$. }
\label{888}
\end{figure}

In the case of $n=2$, many interesting mixed solutions of higher-order soliton, multi-soliton, periodic and double-periodic solutions can be derived for $\lambda_{2}=-\lambda_{1}^{*}=-\alpha_{1}+i\beta_{1}$, $\lambda_{N-2}=-\lambda_{N-3}^{*}=-\alpha_{N-2}+i\beta_{N-2}$ and $\lambda_{N}=-\lambda_{N-1}^{*}=-\alpha_{N-1}+i\beta_{N-1}$.
 Let us take $N=8$ to show the elastic collision of a second-order soliton and a two-soliton solution, and the dynamic evolution diagram is given in Fig. \ref{2j+1+1}. Specially, an elastic collision of a second-order soliton and a velocity resonance two-soliton is obtained when $\beta_{5}^{2}-\alpha_{5}^{2}$=$\beta_{7}^{2}-\alpha_{7}^{2}$ (see Fig. \ref{2j+2gt}), velocity resonance of a second-order soliton and two-soliton is formed when $\beta_{1}^{2}-\alpha_{1}^{2}$=$\beta_{5}^{2}-\alpha_{5}^{2}$=$\beta_{7}^{2}-\alpha_{7}^{2}$ (see Fig. \ref{vc2j+2g}).
\begin{figure}[ht!]
\centering
\subfigure[]{
\label{2j+1+1}
\begin{minipage}[b]{0.3\textwidth}
\includegraphics[width=3.4cm]{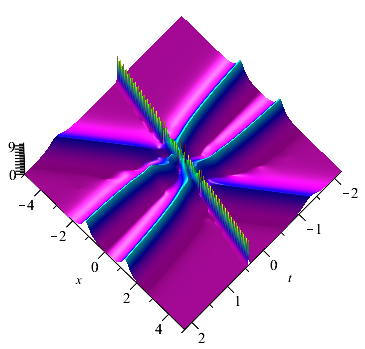}\\
\includegraphics[width=3cm]{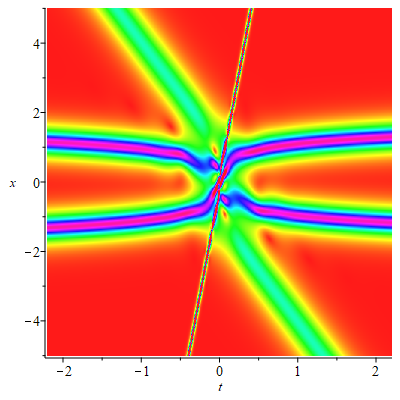}
\end{minipage}}
\subfigure[]{
\label{2j+2gt}
\begin{minipage}[b]{0.3\textwidth}
\includegraphics[width=4cm]{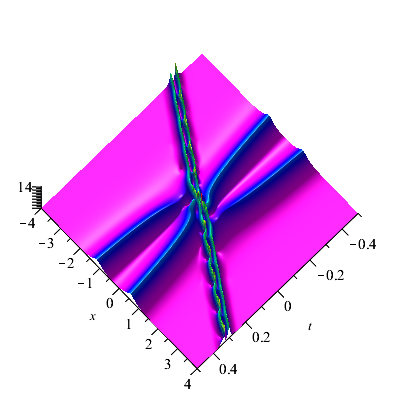}\\
\includegraphics[width=3cm]{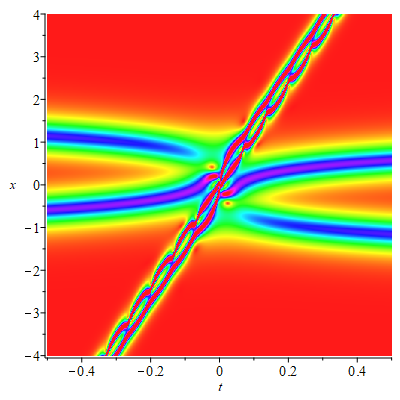}
\end{minipage}}
\subfigure[]{
\label{vc2j+2g}
\begin{minipage}[b]{0.3\textwidth}
\includegraphics[width=4cm]{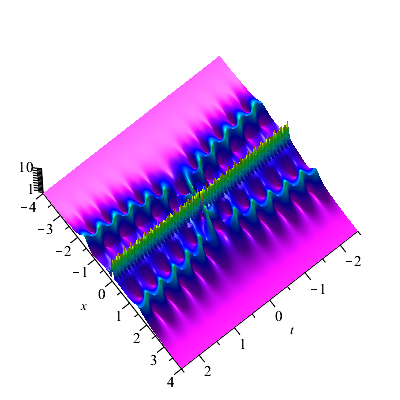}\\
\includegraphics[width=3cm]{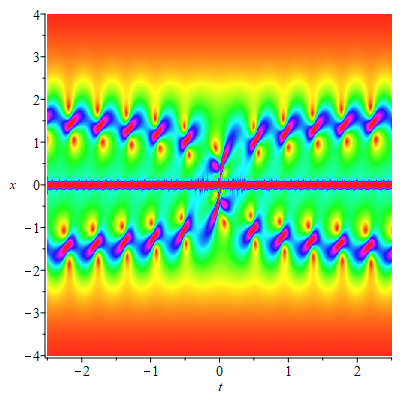}
\end{minipage}}
\caption{(a) Elastic collision of second-order soliton and two-soliton: $\alpha_{1}=\beta_{1}=\alpha_{5}=\alpha_{7}=1$, $\beta_{5}=0.5$ and $\beta_{7}=2$; (b) Elastic collision of second-order soliton and velocity resonance two-soliton: $\alpha_{1}=\beta_{1}=\alpha_{5}=1$, $\beta_{5}=2$, $\alpha_{7}=\sqrt{2}$ and $\beta_{7}=\sqrt{5}$;
(c) Velocity resonance of second-order soliton and two-soliton: $\alpha_{1}=\beta_{1}=1$, $\alpha_{5}=\beta_{5}=0.5$ and $\alpha_{7}=\beta_{7}=2$.}
\end{figure}

In addition, elastic collisions between a second-order soliton and a one-soliton on the periodic can be viewed when $\lambda_{2}=-\lambda_{1}^{*}=-\alpha_{1}+i\beta_{1}$, $\lambda_{6}=-\lambda_{5}^{*}=-\alpha_{5}+i\beta_{5}$,  $\lambda_{7}=i\beta_{7}$ and $\lambda_{8}= i\beta_{8}$.  In particular, when $\beta_{5}^{2}-\alpha_{5}^{2}$=$\beta_{1}^{2}-\alpha_{1}^{2}$, the velocity resonance of a second-order soliton and a one-soliton on the periodic background can be observed. When $\lambda_{2}=-\lambda_{1}^{*}=-\alpha_{1}+i\beta_{1}$ and $\lambda_{k}=i\beta_{k}$ are pure imaginary numbers, we can construct the second-order soliton on the double-periodic background for $k=5$, $6$, $7$ and $8$. The dynamical evolution diagrams of these higher-order solitons on the periodic and double-periodic backgrounds are plotted in Fig. \ref{2j+2pzt}, which obviously reveals that the structures of solitons are greatly influenced by  their propagating directions and positions on the periodic background.

\begin{figure}[ht!]
\centering
\subfigure[]{
\label{2j+1s+p}
\begin{minipage}[b]{0.3\textwidth}
\includegraphics[width=4cm]{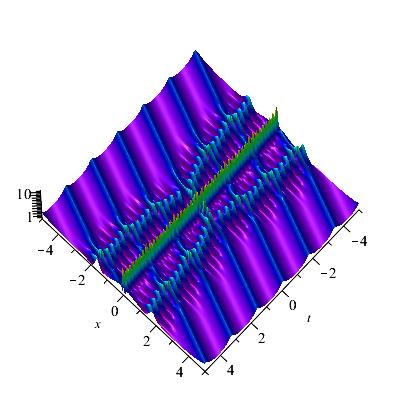}\\
\includegraphics[width=3cm]{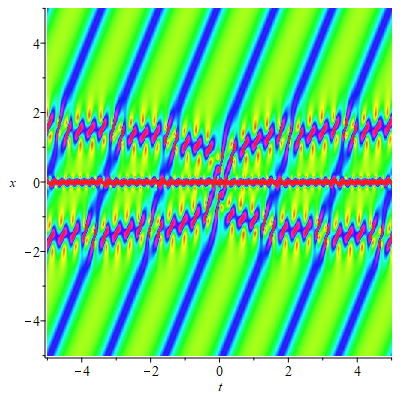}
\end{minipage}}
\subfigure[]{
\label{2j+1t+p}
\begin{minipage}[b]{0.3\textwidth}
\includegraphics[width=4cm]{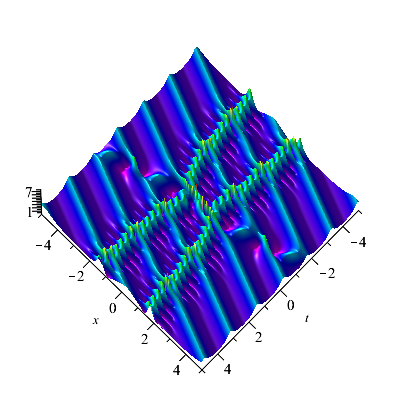}\\
\includegraphics[width=3cm]{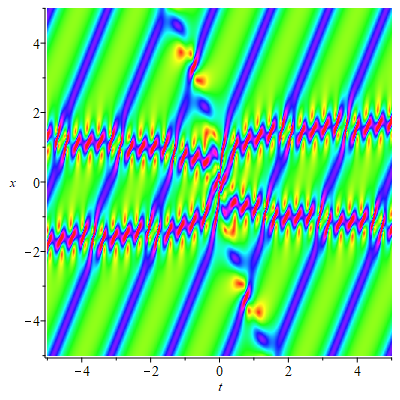}
\end{minipage}}
\subfigure[]{
\label{2j+2p}
\begin{minipage}[b]{0.3\textwidth}
\includegraphics[width=3.2cm]{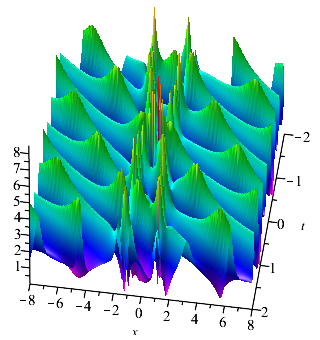}\\
\includegraphics[width=3cm]{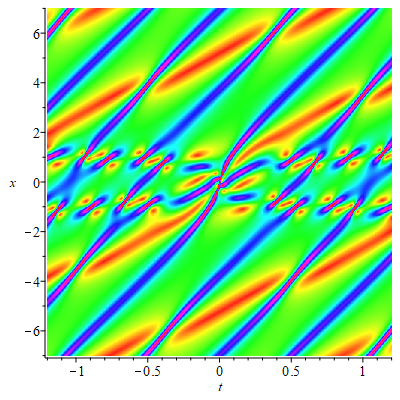}
\end{minipage}}
\caption{(a) Elastic collision of second-order soliton and one soliton on the periodic: $\alpha_{1}=\beta_{1}=\alpha_{5}=\beta_{7}=1$ and $\beta_{5}=\beta_{8}=0.5$;
(b) Velocity resonance of second-order soliton and one-soliton on the periodic background: $\alpha_{1}=\beta_{1}=\beta_{7}=1$, $\alpha_{5}=\beta_{5}=2$ and $\beta_{8}=0.5$;
(c) Second-order soliton on the double-periodic background:  $\alpha_{1}=\beta_{1}=\beta_{7}=1$, $\beta_{5}=0.1$, $\beta_{6}=\sqrt{2}$ and $\beta_{8}=0.5$.}
\label{2j+2pzt}
\end{figure}


{\bf \section{Higher-order hybrid-pattern soliton without and with $n$-periodic background}}

In the previous section, we have investigated the dynamics of the higher-order soliton on the periodic and double periodic backgrounds. It is shown that even though the higher-order soliton is moving on different trajectories, it shares the nearly equal velocity and amplitude. Thanks to the higher-order hybrid-pattern solitons describing the interaction between several types of solitons, they have more abundant and interesting wave structures than higher-order solitons. In this section, we are concentrate on the higher-order hybrid-pattern soliton without and with the periodic backgrounds by the generalized degenerate and semi-degenerate DT.

{\bf \subsection{Generalized degenerate and semi-degenerate DT}}
In order to construct higher-order hybrid-pattern soliton without and with $n$-periodic backgrounds, we  first provide the generalized degenerate DT and generalized semi-degenerate DT formula.
\begin{theorem}{\bf (Generalized degenerate and semi-degenerate DT):}
Letting $$ \lambda_{j} \rightarrow \left\{
                                  \begin{array}{ll}
                                    \lambda_{1}, & \hbox{$2n_{0}<j=odd\leq 2n_{0}+2n_{1}-1$} \\
                                    \lambda_{2}, & \hbox{$2n_{0}<j=even\leq 2n_{0}+2n_{1}$} \\
                                      \lambda_{3}, & \hbox{$2n_{0}+2n_{1} <j=odd\leq2n_{0}+2n_{1}+2n_{2}-1$} \\
                                    \lambda_{4}, & \hbox{$2n_{0}+2n_{1} <j=even\leq2n_{0}+2n_{1}+2n_{2}$} \\
                                   ... \\
                                      \lambda_{2n_{0}-1}, & \hbox{$2n_{0}+2n_{1}+...2n_{\hbar-1} <j=odd\leq2n_{0}+2n_{1}+...2n_{\hbar-1}+2n_{\hbar}-1$} \\
                                    \lambda_{2n_{0}}, & \hbox{$ 2n_{0}+2n_{1}+...2n_{\hbar-1} <j=even\leq2n_{0}+2n_{1}+...2n_{\hbar-1}+2n_{\hbar}$}
                                  \end{array}
                                \right.$$ 
                                and
 $$2\sum_{\ell=0}^{\hbar}n_{\ell}=N-2n,$$
the generalized degenerate DT ($n=0$) and generalized semi-degenerate DT ($n\neq0$) can be derived from the $N$-fold DT formula \eqref{dtf} by the Taylor expansion and determinant calculations. The specific form of the new solution $q_{N}$ is 
\begin{equation}\label{gddtf}
q_{N}=\frac{\left|M'\right|^{2}}{|P'|^{2}}q+2i\frac{|M'||H^{'}|}{|P^{'}|^{2}},
\end{equation}
where
\begin{equation}
\begin{split}
&M'=\left\{
         \begin{array}{ll}
           M'_{jk}, & 1\leq j,k \leq N-2n \\
         M_{jk}, & N-2n<j,k \leq N
         \end{array}
       \right.,\\
&H'=\left\{
         \begin{array}{ll}
           H'_{jk}, & 1\leq j,k \leq N-2n \\
        H_{jk}, & N-2n <j,k \leq N
         \end{array}
       \right.,\\
&P'=\left\{
         \begin{array}{ll}
           P'_{jk}, & 1\leq j,k \leq N-2n \\
         P_{jk}, & N-2n< j,k \leq N
         \end{array}
       \right.,\\
\end{split}
\end{equation}
$$M_{jk}^{\prime}=\lim_{\epsilon \rightarrow 0}\frac{\partial^{n_{j}} M'_{jk}\left(\lambda_{j}+\epsilon\right)}{n_{j}! \partial \epsilon^{n_{j}-1}},$$

$$H_{jk}^{\prime}=\lim_{\epsilon \rightarrow 0}\frac{\partial^{n_{j}} H'_{jk}\left(\lambda_{j}+\epsilon\right)}{n_{j}! \partial \epsilon^{n_{j}-1}},$$

$$P_{jk}^{\prime}=\lim_{\epsilon \rightarrow 0}\frac{\partial^{n_{j}} P'_{jk}\left(\lambda_{j}+\epsilon\right)}{n_{j}! \partial \epsilon^{n_{j}-1}},$$

$$
n_{j}=\left\{ \begin{array}{ll}
           0,& \hbox{$j\leq 2n_{0}$} \\
          \left[\frac{j-2n_{0}+1}{2}\right], & \hbox{$2n_{0} < j \leq 2n_{0}+2n_{1}$} \\
            \left[\frac{j-2n_{0}-2n_{1}+1}{2}\right], & \hbox{$2n_{0}+2n_{1} < j \leq 2n_{0}+2n_{1}+2n_{2}$} \\
           ...\\
 \left[\frac{j-2n_{0}-2n_{1}-...-2n_{\hbar-1}+1}{2}\right], & \hbox{$2n_{0}+2n_{1}+...+2n_{\hbar-1} < j \leq 2n_{0}+2n_{1}+...+2n_{\hbar}$}
          \end{array}
\right..
$$
\end{theorem}
It is remarkable that the proof of this theorem is similar to the proof of the degenerate DT and semi-degenerate DT.

{\bf \subsection{Higher-order hybrid-pattern soliton from the generalized degenerate DT}}

The higher-order hybrid-pattern solitons can be constructed by using the generalized degenerate DT. For instance, an elastic collision of two second-order solitons can be derived if $\lambda_{2}=-\lambda_{1}^{*}=-\alpha_{1}+i\beta_{1}$, $\lambda_{4}=-\lambda_{3}^{*}=-\alpha_{3}+i\beta_{3}$ and $N=8$. It is shown in  Fig. \ref{n422nvr}  that the two second-order solitons are moving along two different trajectories with different velocities and amplitudes. Furthermore, they will become two velocity resonance second-order solitons if $\beta_{3}^{2}-\alpha_{3}^{2}$=$\beta_{1}^{2}-\alpha_{1}^{2}$. Fig. \ref{n422vr} and Fig. \ref{3n422} display two different elastic collisions of two velocity resonance second-order solitons depending on the choices of the parameters.
\begin{figure}[ht!]
\centering
\subfigure[]{
\label{n422nvr}
\begin{minipage}[b]{0.3\textwidth}
\includegraphics[width=3.4cm]{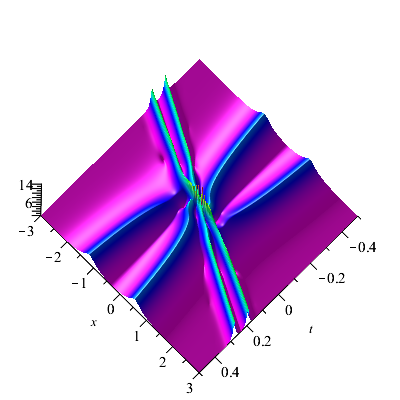}\\
\includegraphics[width=3cm]{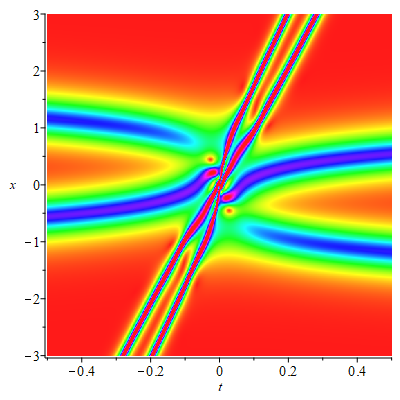}
\end{minipage}}
\subfigure[]{
\label{n422vr}
\begin{minipage}[b]{0.3\textwidth}
\includegraphics[width=3.4cm]{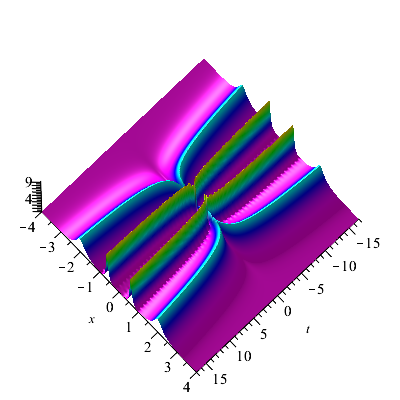}\\
\includegraphics[width=3cm]{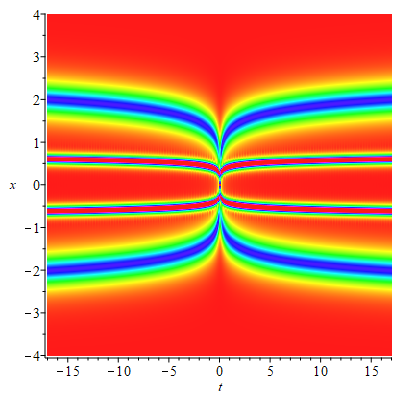}
\end{minipage}}
\subfigure[]{
\label{3n422}
\begin{minipage}[b]{0.3\textwidth}
\includegraphics[width=3.4cm]{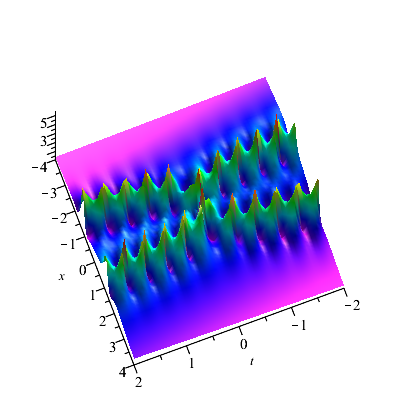}\\
\includegraphics[width=3cm]{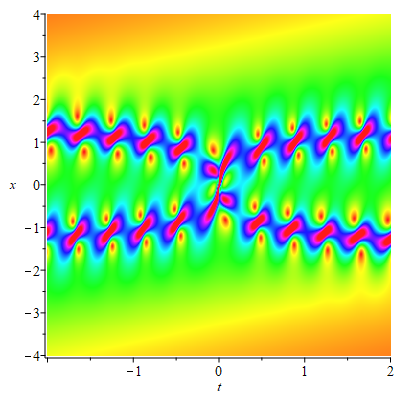}
\end{minipage}}
\caption{(a) Elastic collision of two second-order solitons:
 $\alpha_{1}=\beta_{1}=\alpha_{3}=1$ and $\beta_{3}=2$; (b) Elastic collision of two  velocity resonance second-order solitons:  $\alpha_{1}=\beta_{1}=1$ and $\alpha_{3}=\beta_{3}=2$;  (c) Elastic collision of two second-order velocity resonance solitons:  $\alpha_{1}=\beta_{1}=0.1$ and $\alpha_{3}=\beta_{3}=0.2$. }
\end{figure}

{\bf \subsection{Higher-order hybrid-pattern soliton on the $n$-periodic background from the generalized semi-degenerate DT}}

By means of the generalized semi-degenerate DT, the higher-order soliton on the $n$-periodic background can be obtained. As an application, here we just give the results for the case of $n=1$. Fixing $N=10$, the  elastic collision of two second-order solitons on the periodic background is constructed for $\lambda_{2}=-\lambda_{1}^{*}=-\alpha_{1}+i\beta_{1}$,  $\lambda_{4}=-\lambda_{3}^{*}=-\alpha_{3}+i\beta_{3}$, $\lambda_{9}=i\beta_{9}$ and $\lambda_{10}=i\beta_{10}$, as exhibited in Fig. \ref{q10pb}). In particular, when $\beta_{3}^{2}-\alpha_{3}^{2}$=$\beta_{1}^{2}-\alpha_{1}^{2}$, velocity resonance of two second-order solitons on the periodic background can be formed. As can be seen from the dynamic evolution diagram in Fig. \ref{by1}, higher-order solitons with different velocities and directions show completely different profiles under the influence of the periodic background, and the periodic background will gradually weaken, which is consistent with the wave phenomenon in nature.  Fig. \ref{01010202t8x16} is the local structure  of Fig. \ref{by1}. Fig. \ref{1122} describes a different velocity resonance of two second-order solitons on the periodic background. It is observed that the propagation patterns of solitons change greatly with the different periods of the periodic backgrounds, deciding on the choices of the parameters.

\begin{figure}[ht!]
\centering
\subfigure[]{
\label{q10pb}
\begin{minipage}[b]{0.2\textwidth}
\includegraphics[width=3cm]{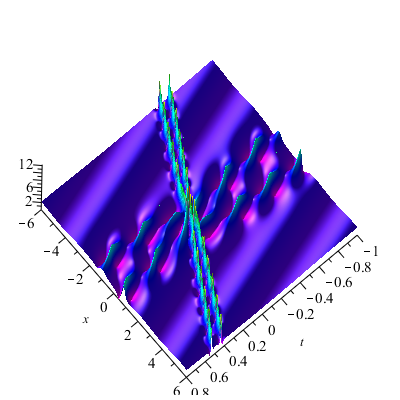}\\
\includegraphics[width=2.4cm]{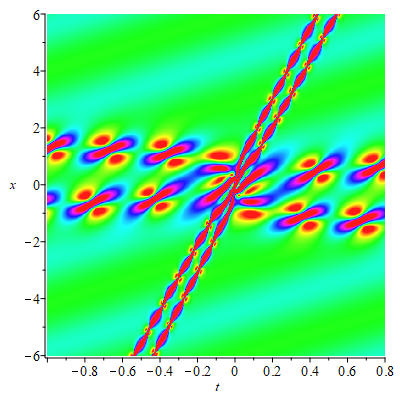}
\end{minipage}}
\subfigure[]{
\label{01010202t8x16}
\begin{minipage}[b]{0.2\textwidth}
\includegraphics[width=3cm]{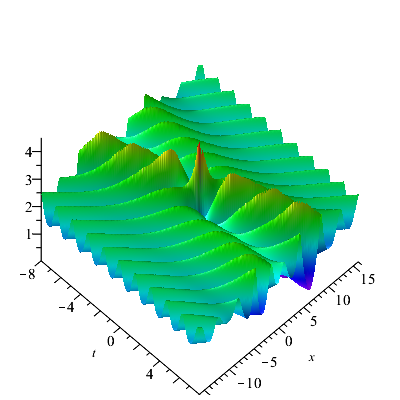}\\
\includegraphics[width=2.4cm]{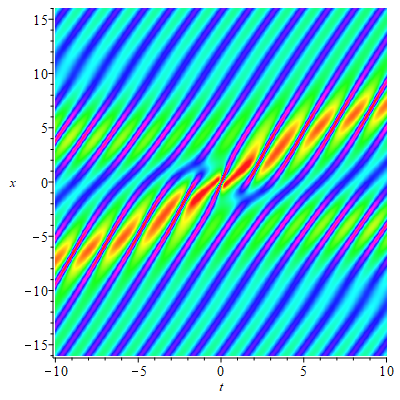}
\end{minipage}}
\subfigure[]{
\label{by1}
\begin{minipage}[b]{0.2\textwidth}
\includegraphics[width=3cm]{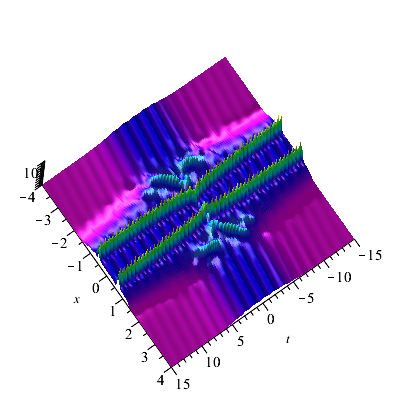}\\
\includegraphics[width=2.4cm]{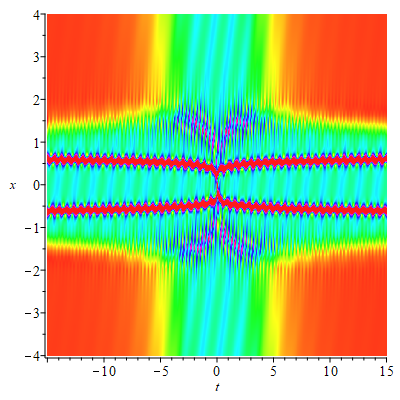}
\end{minipage}}
\subfigure[]{
\label{1122}
\begin{minipage}[b]{0.2\textwidth}
\includegraphics[width=3cm]{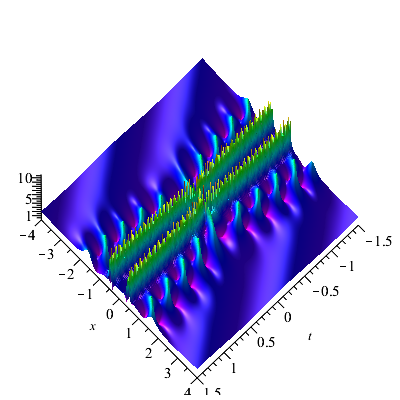}\\
\includegraphics[width=2.4cm]{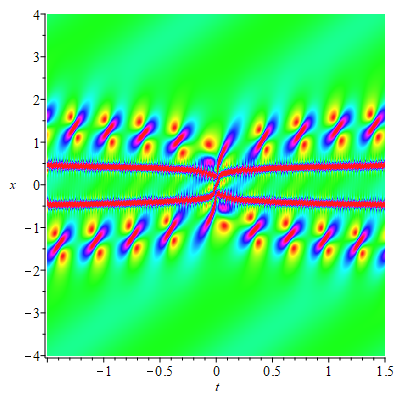}
\end{minipage}}
\caption{(a) Elastic collision of two second-order solitons on the periodic background:
 $\alpha_{1}=\beta_{1}=\alpha_{3}=\beta_{10}=1$, $\beta_{3}=2$ and $\beta_{9}= 0.1$;
(b) Velocity resonance of two second-order solitons on the periodic background:
 $\alpha_{1}=\beta_{1}=\beta_{9}=0.1$, $\alpha_{3}=\beta_{3}=0.2$ and $\beta_{10}=1$;
 (c) and (d) Velocity resonance of two second-order solitons on the periodic background:
  $\alpha_{1}=\beta_{1}=\beta_{10}=1$, $\alpha_{3}=\beta_{3}=2$ and $\beta_{9}= 0.1$.}
\label{dtx}
\end{figure}

{\bf  \section{ Extension to the reverse-space-time DNLS equation} }

In recent years, nonlocal nonlinear equations have been studied a lot from different viewpoints and perspectives \cite{lm-2015-pre,ablowitz-tmp-2018,wmm-nd-2021}, owing to their physical applications in Bose-Einstein condensate \cite{ptbose-pra-2012}, unconventional systems of magnetics \cite{ta-pra-2016}, and so on \cite{LRC-pea-2011,jky-pre-2018}. For the DNLS equation, three types nonlocal extensions, namely, the $PT$-symmetric, the reverse-time and the reverse-space-time  DNLS equations, have been realized by Ablowitz and Musslimani \cite{ab-prl-2013,ab-studies-2017}. The reverse-space-time DNLS equation 
\begin{equation}\label{ndnls}
 iq_{t}-q_{x x}+i(q^{2}q(-x,-t))_{x}=0
\end{equation}
is a reduction of the KN system \eqref{kns} under $r(x,t)=-q(-x,-t)$. The DT for \eqref{ndnls} can also be reduced from that of the KN system.
 
Interestingly, the solutions of Eq. \eqref{ndnls} derived by the $N$-fold DT formula can be reduced to those of the DNLS equation \eqref{dnls}. Therefore, we can extend the results of the DNLS equation to the reverse-space-time DNLS.  In addition, we also present its modulational instability (MI) analysis.

{\bf \subsection{Symmetry condition of eigenfunctions}}

The eigenfunctions of the reverse-space-time DNLS equation have the relation $\phi(x, t; \lambda_{j}) =\varphi(-x, -t; \lambda_{j})$ for every $\lambda_{j}$, which has quite different symmetry property compared with the DNLS equation and the detailed proof is presented in the lemma  below.
 
\begin{lemma}\label{linearconstraint}

{\bf (Symmetry relation of the eigenfunctions for the reverse-space-time DNLS equation):}
For Eq. \eqref{ndnls}, the eigenfunctions admit the symmetry condition
\begin{equation}\label{sc}
\phi(x, t; \lambda_{j}) =\varphi(-x, -t; \lambda_{j}).
\end{equation}
\end{lemma}
\begin{proof}  From the $x$ part of the Lax pair \eqref{xlax}, one has
$$
\left(\begin{array}{l}
\phi(x, t; \lambda_{j})\\
\varphi(x, t; \lambda_{j})
\end{array}\right)_{x}=\left(\begin{array}{cc}
i\lambda_{j}^{2} & q(x, t)\lambda_{j} \\
-q(-x, -t)\lambda_{j} & -i\lambda_{j}^{2}
\end{array}\right)\left(\begin{array}{l}
\phi(x, t; \lambda_{j}) \\
\varphi(x, t; \lambda_{j})
\end{array}\right).
$$
Let $x\rightarrow -x$, $t\rightarrow -t$,  then
$$
\left(\begin{array}{l}
\varphi(-x, -t; \lambda_{j})\\
\phi(-x, -t; \lambda_{j})
\end{array}\right)_{x}=\left(\begin{array}{cc}
i\lambda_{j}^{2} & q(x, t)\lambda_{j} \\
-q(-x, -t)\lambda_{j} & -i\lambda_{j}^{2}
\end{array}\right)\left(\begin{array}{l}
\varphi(-x, -t; \lambda_{j}) \\
\phi(-x, -t; \lambda_{j})
\end{array}\right).
$$
So
$$
\left(\begin{array}{l}
\phi(x, t; \lambda_{j})\\
\varphi(x, t; \lambda_{j})
\end{array}\right)  \text{and} \left(\begin{array}{l}
\varphi(-x, -t; \lambda_{j})\\
\phi(-x, -t; \lambda_{j})
\end{array}\right)$$
satisfy the same spectral problem. Taking a similar procedure, the symmetry property also holds for the $t$ part of the Lax pair. That means the symmetry relation of the eigenfunctions for the reverse-space-time DNLS equation is $\phi(x, t; \lambda_{j}) =\varphi(-x, -t; \lambda_{j})$ for every $\lambda_{j}$.
\end{proof}

First, we take the simplest case $\lambda_{1}=\alpha_{1}+i\beta_{1}$. Then the solution $q[1]$ for the reverse-space-time DNLS equation is 
\begin{equation}
\begin{split}
q[1]&=(2i\alpha_{1}-2\beta_{1}) e^{iA+B},\\
A&=-2(2\alpha_{1}^{4}t-12\alpha_{1}^{2}\beta_{1}^{2}t+2\beta_{1}^{4}t+\alpha_{1}^{2}x-\beta_{1}^{2}x),\\
B&=2(8\alpha_{1}^{3}\beta_{1}t-8\alpha_{1}\beta_{1}^{3}t+2\alpha_{1}\beta_{1}x).
\end{split}
\end{equation}
It is noted that when $\alpha_{1}\beta_{1}=0$, we have $|q[1]|^{2}=4(\alpha_{1}^{2}+\beta_{1}^{2})$, so it is a plane wave with the height of $2\sqrt{\alpha_{1}^{2}+\beta_{1}^{2}}$. When $\alpha_{1}\beta_{1}\neq0$, we have $|q[1]|^{2}=4(\alpha_{1}^{2}+\beta_{1}^{2})e^{8\alpha_{1}\beta_{1}(4\alpha_{1}^{2}t-4\beta_{1}^{2}t+x)}$, and hence, it is an exponentially growing curved surface wave.\\

Now, let us take $\lambda_{1}=\alpha_{1}+i\beta_{1}$, and $\lambda_{2}=\alpha_{2}+i\beta_{2}$. Then the explicit expression of $q[2]$ is 
\begin{equation} \label{0q2}
q[2]={\frac {2i\zeta_{1}e^{iA_{1}}B
[\zeta_{1}B\sin{A_{4}}\sin{A_{3}}+\zeta_{2}\sin(A_{3}-A_{4})
+\zeta_{1}\cos{A_{4}}\cos{A_{3}}]}{[\zeta_{2}\sin{(A_{3}-A_{4})}-\zeta_{1}\cos{ (A_{3}+ A_{4})}]^{2}}},
\end{equation}
where
\begin{equation}
\begin{split}
&B=\cosh{ A_{2}}-\sinh{A_{2}},\\
&A_{1}=(2\alpha_{1}^{4}-12\alpha_{1}^{2}\beta_{1}^{2}+2\alpha_{2}^{4}-12\alpha_{2}^{2}\beta_{2}^{2}+2\beta_{1}^{4}+2\beta_{2}^{4})t
-(-\alpha_{1}^{2}-\alpha_{2}^{2}+\beta_{1}^{2}+\beta_{2}^{2})x,\\
&A_{2}=(8\alpha_{1}^{3}\beta_{1}-8\alpha_{1}\beta_{1}^{3}+8\alpha_{2}^{3}\beta_{2}-8\alpha_{2}\beta_{2}^{3})t
+2x(\alpha_{1}\beta_{1}+\alpha_{2}\beta_{2}),\\
&A_{3}=2\alpha_{2}^{4}t+8i\alpha_{2}^{3}\beta_{2}t+(-12\beta_{2}^{2}t+x)\alpha_{2}^{2}-8i(\beta_{2}^{2}t-\frac{1}{4}x)\beta_{2}\alpha_{2}
+2\beta_{2}^{4}t-\beta_{2}^{2}x,\\
&A_{4}=2\alpha_{1}^{4}t+8i\alpha_{1}^{3}\beta_{1}t+(-12\beta_{1}^{2}t+x)\alpha_{1}^{2}-8i(\beta_{1}^{2}t
-\frac{1}{4}x)\beta_{1}\alpha_{1}+2\beta_{1}^{4}t-\beta_{1}^{2}x,\\
&\zeta_{1}=\alpha_{1}+i\beta_{1}-\alpha_{2}-i\beta_{2},\quad \zeta_{2}= i\alpha_{1}+i\alpha_{2}-\beta_{1}-\beta_{2}.
\end{split}
\end{equation}

When $\alpha_{2}=\alpha_{1}$, $\beta_{2}=-\beta_{1}$ or $\alpha_{2}=-\alpha_{1}$, $\beta_{2}=\beta_{1}$, then $A_{2}=0$ and $A_{3}+A_{4}=0$ in Eq. \eqref{0q2}. Namely, soliton solutions can be constructed by taking $\lambda_{2}=\pm\lambda_{1}^{*}=\pm(\alpha_{1}-i\beta_{1})$. 
When $\alpha_{1}=\alpha_{2}=0$ or $\beta_{1}=\beta_{2}=0$, then $A_{2}=0$ in Eq. \eqref{0q2}. It means $q[2]$ can represent periodic solutions when $\lambda_{1}$ and $\lambda_{2}$ are pure imaginary numbers or real numbers.  

Following, similar to the above method of studying the solution of the DNLS equation, we can extend the results related to the DNLS equation obtained in this paper to the reverse-space-time DNLS equation.

{\bf  \subsection{Modulational Instability} }
The MI analysis is carried out to show the conditions of the generating MI and modulational stability (MS) regions for the reverse-space-time DNLS equation. It is easy to verify that Eq. \eqref{ndnls} has the plane wave solution $q_{0}(x,t)=ce^{i(kx+wt)}$, where the background frequency $w=-c^{2}k+k^{2}$, $c$ is the amplitude and $k$ is the wave number.
 According to the MI theory, the perturbation solution is of the form $q_{1}(x,t)=(c+P)e^{i(kx+wt)}$, where $P= m\exp(i(K x+\Omega t))+n\exp(-i(K x+\Omega t))$, $K$ is the  disturbance wave number and $\Omega$ is the  disturbance  frequency. Substituting $q_{1}(x,t)$ into Eq. \eqref{ndnls}, we get a system of linear homogeneous equations for the small parameters $m$ and $n$
\begin{equation}
\begin{split}
 ((k+2K)c^{2}+(-K^{2}-2kK+\Omega))m+(k+K)c^{2}n=0,\\
(k-K)c^{2} m+((k-2K)c^{2}+2kK-K^{2}-\Omega)n=0,
\end{split}
\end{equation}
which gives rise to the dispersion relation 
\begin{equation}
\Omega =(-2c^{2}+2k\pm\sqrt{c^{4}-2c^{2}k+K^{2}})K.
\end{equation}
We conclude that the MI depends on the values of the amplitude $c$, the wave number $k$ and the perturbation wave number $K$. From the above dispersion relation, it is seen that when $c^{4}-2c^{2}k+K^{2}\geq0$, the  frequency  $\Omega$ is real at any value of the wave number $K$, otherwise, $\Omega$ becomes complex and hence disturbance will grow exponentially in time. The power gain is obtained as
$$ G(K)=\text{Im}(\Omega)=\text{Im}(|K|\sqrt{c^{4}-2c^{2}k+K^{2}}),$$ which represents the MI gain when $c^{4}-2c^{2}k+K^{2}<0$. The imaginary part makes the perturbation function $P$ increase exponentially and destroy the stability of the system, and this instability is a condition for the existence of a rogue wave. There exists two distinctive MI and MS regions. In the region of $c^{4}-2c^{2}k+K^{2}<0$,  MI exists, otherwise, MS region appears.

Let $c=1$, then the MI arises at $k>0.5$. Fig. \ref{MI} shows the gain at    three different plane-wave numbers $k$, and Fig. \ref{dMI} gives the gain function. Actually, the  gain function  $G(K)$ is an even function of $K$, so the MI figure is symmetric about the line $K=0$. 

\begin{figure}[ht!]
\centering
\subfigure[]{
\label{MI}
\includegraphics[width=4.2cm]{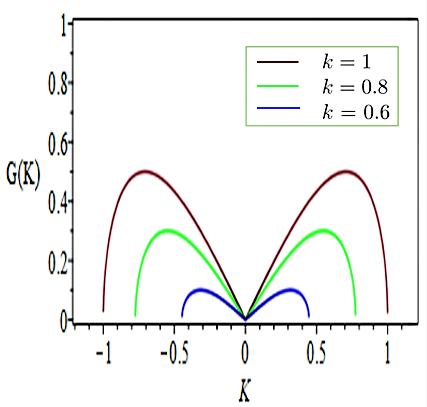}}
\subfigure[]{
\label{dMI}
\includegraphics[width=4cm]{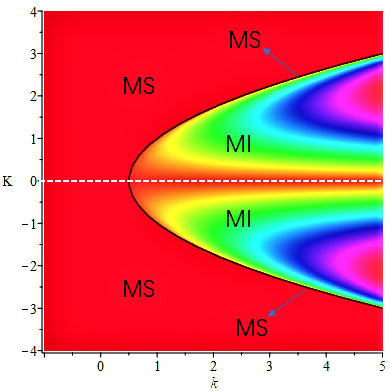}}
\caption{(a) Gain $G(k)$ at three different plane-wave numbers $k$; (b) Gain function.}
\label{fig1}
\end{figure}

\section{Conclusion and discussion}
\par
\ \ \ \
In this paper, the $N$-fold DT formula of the KN system was expressed in a more concise form for any $N$. Then the degenerate and semi-degenerate DT, generalized degenerate and semi-degenerate DT formulae were obtained based on the degenerate $N$-fold DT formula. The multi-soliton solutions, $n$-periodic wave solutions and some mixed solutions were generated from the even-fold DT with different parameter constraints. Using the degenerate and semi-degenerate DT, we obtained the mixed solutions of multi-solitons, $n$-periodic solutions and higher-order solitons. The more general and complex higher-order hybrid-pattern solitons without and with $n$-periodic backgrounds were considered via constructing the generalized DT.  
 
The dynamic evolution of single periodic and double-periodic waves looks very regular. But the $n$-periodic solution ($n>2$) shows a more complex structure which has peaks with different amplitudes and sizes. It is due to the elastic collision of periodic solutions with different directions and velocities. The local structure of soliton on the periodic-background has a single peak with two caves which is similar to the rogue waves. This gives us an idea to construct rogue wave solutions, but the feasibility remains to be proved. It is amazing that our results of the DNLS equation can be extended directly to the reverse-space-time DNLS equation. We also believe that the method can be applied to many other physical systems, such as the Gerdjikov-Ivanov equation.
 
The results obtained in this investigation are of importance for understanding the various soliton phenomena in many fields governed by the local and nonlocal nonlinear dynamical systems, such as nonlinear optics, Bose-Einstein condensates and other relevant fields.

\section*{Acknowledge}
This work was supported by National Natural Science Foundation of China (No.12175069), Global Change Research Program of China (No.2015CB953904), Science and Technology Commission of Shanghai Municipality (No.18dz2271000 and No.21JC1402500).

\section*{Declarations}

\section*{Conflict of interests}
The authors declare that there is no conflict of interests regarding the publication of this paper.

\section*{Data availability statement}
 All date generated or analysed during this study are included in this published article.




\begin{thebibliography}{}
\bibitem{rogister1971} A. Rogister, Phys. Fluids, {\bf 14}, 2733  (1971).
\bibitem{Mjlhus1974} E. Mj$\phi$lhus, Department of Applied Mathematics, University of Bergen, Rep. 48. (1974).
\bibitem{Mjlhus1976} E. Mj$\phi$lhus, J. Plasma Phys. {\bf 16}, 321 (1976).
\bibitem{Mio1976} K. Mio, T. Ogino, K. Minami, and S. Takeda, J. Phys. Soc. Japan, {\bf 41}, 265 (1976).
\bibitem{Ichikawa1977} Y. H. Ichikawa and S. Watanabe, J. Phys. thior. appl. Supp. fasc. {\bf 12}, 6 (1977).
\bibitem{Spatschek1977} K.H. Spatchek, P.K. Shukla, and  M.Y. Yu, Nucl. Fus. {\bf 18}, 290 (1977).
\bibitem{iykk-jpsj-1980} Y. Ichikawa, K. Konno, M. Wadati, and Sanuki, H. J. Phys. Soc. Jpn. {\bf 48}, 279 (1980).
\bibitem{cxj-pre-2004} X.J. Chen and W.K. Lam, Phys. Rev. E {\bf 69}, 066604 (2004).
\bibitem{kam-pla-1998} A.M. Kamchatnov, S.A. Darmanvan, and F. Lederer, Phys. Lett. A {\bf 245}, 259 (1998).
\bibitem{plasma2} M.S. Ruderman, J. Plasma Phys. {\bf 67}, 271 (2002).
\bibitem{7} N. Tzoar and  M. Jain, Phys. Rev. A {\bf 23}, 1266 (1981).
\bibitem{8} D. Anderson and M. Lisak, Phys. Rev. A {\bf 27}, 1393 (1983).
\bibitem{kn-jmp-1978} D.J. Kaup and A.C. Newell, J. Math. Phys. {\bf 19}(4), 798 (1978).
\bibitem{kj-jpsj-1999} K.J. Imai, J. Phys. Soc. Japan. {\bf 68} (2), 355 (1999).
\bibitem{JPAXSW-2011} S.W. Xu, J.S. He, and  L.H. Wang, J. Phys. A-Math. Theor. {\bf 44}, 6629 (2011).
 \bibitem{gbl-sam-2012} B.L. Guo,  L.M. Ling, and  Q.P. Liu, Stud. Appl. Math. {\bf 130}, 317 (2012).
   \bibitem{akhmediev-1987-tmp}  N.N. Akhmediev,  V.M. Eleonskii, and N.E. Kulagin, Theor. Math. Phys. {\bf 72}, 809 (1987).
\bibitem{hxr-2012-pre} X.R. Hu, S.Y. Lou, and Y. Chen, Phys. Rev. E {\bf 85}, 056607 (2012).
\bibitem{cjb2018non} J.B. Chen and D.E. Pelinovsky, Nonlinearity {\bf 31}, 1955 (2018).
 \bibitem{lwhjs2018} W. Liu,  Y.S. Zhang, and J.S. He, Rom. Rep. Phys. {\bf 70}, 106 (2018).
\bibitem{zhj-nd-2021} H.J. Zhou and Y. Chen, Nonlinear Dyn. 
 {\bf 106}, 3437 (2021).  
  \bibitem{txy-pre-2002} X.Y. Tang, S.Y. Lou, and Y. Zhang, Phys. Rev. E {\bf 66}(4), 046601 (2002).
\bibitem{lm-2015-pre} M. Li and T. Xu, Phys. Rev. E {\bf 91}, 033202 (2015).
\bibitem{ablowitz-tmp-2018} M.J. Ablowitz, B.F. Feng, X.D. Luo, and Z.H. Musslimani, Theor. Math. Phys. {\bf 196}, 1241 (2018).
\bibitem{wmm-nd-2021}  M.M. Wang and Y. Chen, Nonlinear Dyn. {\bf 104}, 2621 (2021).
\bibitem{ptbose-pra-2012} H. Cartarius and G. Wunner, Phys. Rev. A 
{\bf 86}, 013612 (2012).
\bibitem{ta-pra-2016} T. A. Gadzhimuradov and A. M. Agalarov, Phys. Rev. A {\bf 93}, 062124 (2016).
\bibitem{LRC-pea-2011} J. Schindler, A. Li, M. C. Zheng, F. M. Ellis, and T. Kottos, Phys. Rev. A {\bf 84}, 040101 (2011).
\bibitem{jky-pre-2018} J.K Yang, Phys. Rev. E {\bf 98}, 042202 (2018).
\bibitem{ab-prl-2013} M.J. Ablowitz and Z.H. Musslimani, Phys. Rev. Lett. {\bf 110}, 064105 (2013).
\bibitem{ab-studies-2017} M.J. Ablowitz and Z.H. Musslimani, Stud. Appl. Math. {\bf 139}(1), 7 (2017).
\end{thebibliography}


\end{document}